\documentclass[envcountsame,runningheads,11pt]{llncs}
\usepackage{a4wide}
\usepackage{xspace}
\usepackage{amsmath}
\usepackage{amssymb}
\usepackage{graphicx}
\DeclareFontFamily{U}{mathb}{\hyphenchar\font45}
\DeclareFontShape{U}{mathb}{m}{n}{
      <5> <6> <7> <8> <9> <10> gen * mathb
      <10.95> mathb10 <12> <14.4> <17.28> <20.74> <24.88> mathb12
      }{}
\DeclareSymbolFont{mathb}{U}{mathb}{m}{n}
\DeclareMathSymbol{\precneq}{3}{mathb}{"AC}
\newcommand{\myurl}[1]{\url{{\rm\texttt{#1}}}\xspace}
\newcommand{\reducneq}{\precneq}
\newcommand{\reduceq}{\preceq}

\newcommand{\reduceqt}{\reduceq_{\rm t}}

\newcommand{\IR}{\mathbb{R}}
\newcommand{\IRc}{\mathbb{R}_{\rm c}}

\newcommand{\ID}{\mathbb{D}}
\newcommand{\IQ}{\mathbb{Q}}
\newcommand{\IP}{\mathbb{P}}
\newcommand{\IC}{\mathbb{C}}
\newcommand{\IK}{\mathbb{K}}
\newcommand{\IZ}{\mathbb{Z}}

\newcommand{\IN}{\mathbb{N}}
\newcommand{\IS}{\mathbb{S}}
\newcommand{\Sierp}{\IS}
\newcommand{\Sierpinski}{Sierpi\'{n}ski\xspace}
\newcommand{\MLPO}{\text{MLPO}}
\newcommand{\Kol}{\text{C}}
\newcommand{\Kolu}{\text{C}_{\rm u}}
\newcommand{\Adic}{\operatorname{Adic}}
\newcommand{\dom}{\operatorname{dom}}
\newcommand{\range}{\operatorname{range}}
\newcommand{\bin}{\operatorname{bin}}
\newcommand{\rank}{\operatorname{rank}}
\newcommand{\lspan}{\operatorname{lspan}}
\newcommand{\norm}{\operatorname{norm}}
\newcommand{\kernel}{\operatorname{kernel}}
\newcommand{\LinEq}{\operatorname{LinEq}}
\newcommand{\Diag}{\operatorname{Diag}}
\newcommand{\Lev}{\operatorname{Lev}}
\newcommand{\LEV}{\operatorname{LEV}}
\newcommand{\EVec}{\operatorname{EVec}}
\newcommand{\Select}{\operatorname{Select}}
\newcommand{\Intermed}{\operatorname{Intermed}}

\newcommand{\Hausd}{\text{--}}

\newcommand{\ext}{\operatorname{ext}}
\newcommand{\extchull}{\operatorname{extchull}}
\newcommand{\id}{\operatorname{id}}
\newcommand{\cf}[1]{\mathbf{1}_{#1}}

\newcommand{\Tnull}{\ensuremath{\text{T}_0}}
\newcommand{\Tone}{\ensuremath{\text{T}_1}}
\newcommand{\Tthree}{\ensuremath{\text{T}_3}}

\DeclareMathOperator*{\ulim}{ulim}
\newcommand{\calO}{\mathcal{O}}
\newcommand{\calF}{\mathcal{F}}

\newcommand{\calS}{\mathcal{S}}
\newcommand{\calI}{\mathcal{I}}
\newcommand{\calA}{\mathcal{A}}
\newcommand{\calB}{\mathcal{B}}

\newcommand{\person}[1]{\textsc{#1}}

\newcommand{\mycite}[2]{{\rm\cite[\textsc{#1}]{#2}}}
\newcommand{\myto}{\!\to\!}
\newcommand{\toto}{\rightrightarrows}

\newcommand{\ball}{B}
\newcommand{\closure}[1]{\overline{#1}}

\newcommand{\mycard}{\mathfrak{C}_{\rm t}}
\newcommand{\mycomp}{\mathfrak{C}_{\rm c}}
\newcommand{\frakc}{\mathfrak{c}}

\newcommand{\myrho}{\rho}
\newcommand{\myl}{{\scriptscriptstyle<}}
\newcommand{\myg}{{\scriptscriptstyle>}}
\newcommand{\myrhol}{\myrho_{\raisebox{0.2ex}{$\myl$}}}
\newcommand{\myrhog}{\myrho_{\raisebox{0.2ex}{$\myg$}}}

\newcommand{\psiL}[1]{\psi^{\hspace*{-0.7pt}#1}_{\!\raisebox{0.2ex}{$\myl$}}}
\newcommand{\psiG}[1]{\psi^{\hspace*{-0.7pt}#1}_{\!\raisebox{0.2ex}{$\myg$}}}

\newcommand{\psidl}{\psiL{d}}
\newcommand{\psidg}{\psiG{d}}

\newcommand{\chull}{\operatorname{chull}}
\newcommand{\Lim}{\operatorname{Lim}}

\newcommand{\KleeneS}{\Sigma}

\newcommand{\BorelS}{\mathbf{\Sigma}}
\newcommand{\BorelP}{\mathbf{\Pi}}

\newcommand{\Card}{\operatorname{Card}}
\newcommand{\Gdelta}{\text{G}_{\delta}}
\newcommand{\Fsigma}{\text{F}_{\sigma}}

\newcommand{\COMMENTED}[1]{}

\spnewtheorem{observation}[theorem]{Observation}{\bfseries}{\itshape}
\spnewtheorem{fact}[theorem]{Fact}{\bfseries}{\itshape}
\spnewtheorem{myclaim}[theorem]{Claim}{\bfseries}{\itshape}
\spnewtheorem{scholium}[theorem]{Scholium}{\bfseries}{\itshape}
\spnewtheorem{myexample}[theorem]{Example}{\bfseries}{\itshape}
\spnewtheorem{myremark}[theorem]{Remark}{\bfseries}{\itshape}
\spnewtheorem{myquestion}[theorem]{Question}{\bfseries}{\itshape}

\begin{document}
\setcounter{secnumdepth}{3}
\setcounter{tocdepth}{3}
\title{\vspace*{-5ex}Real Computation with Least Discrete Advice: \\
A Complexity Theory of Nonuniform Computability}
\titlerunning{Real Computation with Least Discrete Advice:
A Complexity Theory of Nonuniform Computability}
\author{\vspace*{-1ex}Martin Ziegler\thanks{Supported by \textsf{DFG} grant \texttt{Zi\,1009/1-2}.
The author wishes to thank 
\textsc{Andrej Bauer, Vasco Brattka, Mark Braverman, Peter Hertling},
and \textsc{Arno Pauly} for their helpful advice (pun) during \textsf{CCA2009}
which spurred Sections~\ref{s:Fractional} and \ref{s:Topology}}}
\authorrunning{Martin Ziegler}
\institute{Technische Universit\"{a}t Darmstadt, GERMANY}
\date{}
\makeatletter
\renewcommand\maketitle{\newpage
  \refstepcounter{chapter}%
  \stepcounter{section}%
  \setcounter{section}{0}%
  \setcounter{subsection}{0}%
  \setcounter{figure}{0}
  \setcounter{table}{0}
  \setcounter{equation}{0}
  \setcounter{footnote}{0}%
  \begingroup
    \parindent=\z@
    \renewcommand\thefootnote{\@fnsymbol\c@footnote}%
    \if@twocolumn
      \ifnum \col@number=\@ne
        \@maketitle
      \else
        \twocolumn[\@maketitle]%
      \fi
    \else
      \newpage
      \global\@topnum\z@   
      \@maketitle
    \fi
    \thispagestyle{empty}\@thanks
    \def\\{\unskip\ \ignorespaces}\def\inst##1{\unskip{}}%
    \def\thanks##1{\unskip{}}\def\fnmsep{\unskip}%
    \instindent=\hsize
    \advance\instindent by-\headlineindent
    \if@runhead
       \if!\the\titlerunning!\else
         \edef\@title{\the\titlerunning}%
       \fi
       \global\setbox\titrun=\hbox{\small\rm\unboldmath\ignorespaces\@title}%
       \ifdim\wd\titrun>\instindent
          \typeout{Title too long for running head. Please supply}%
          \typeout{a shorter form with \string\titlerunning\space prior to
                   \string\maketitle}%
          \global\setbox\titrun=\hbox{\small\rm
          Title Suppressed Due to Excessive Length}%
       \fi
       \xdef\@title{\copy\titrun}%
    \fi
    \if!\the\tocauthor!\relax
      {\def\and{\noexpand\protect\noexpand\and}%
      \protected@xdef\toc@uthor{\@author}}%
    \else
      \def\\{\noexpand\protect\noexpand\newline}%
      \protected@xdef\scratch{\the\tocauthor}%
      \protected@xdef\toc@uthor{\scratch}%
    \fi
    \if@runhead
       \if!\the\authorrunning!
         \value{@inst}=\value{@auth}%
         \setcounter{@auth}{1}%
       \else
         \edef\@author{\the\authorrunning}%
       \fi
       \global\setbox\authrun=\hbox{\small\unboldmath\@author\unskip}%
       \ifdim\wd\authrun>\instindent
          \typeout{Names of authors too long for running head. Please supply}%
          \typeout{a shorter form with \string\authorrunning\space prior to
                   \string\maketitle}%
          \global\setbox\authrun=\hbox{\small\rm
          Authors Suppressed Due to Excessive Length}%
       \fi
       \xdef\@author{\copy\authrun}%
       \markboth{\@author}{\@title}%
     \fi
  \endgroup
  \setcounter{footnote}{\fnnstart}%
  \clearheadinfo}
\makeatother

\maketitle
\def\thefootnote{\fnsymbol{footnote}}
\addtocounter{footnote}{3}
\begin{abstract}
It is folklore particularly in numerical and computer sciences
that, instead of solving some general problem $f:A\to B$, 
additional structural information about the input $x\in A$
(that is any kind of promise that $x$ belongs 
to a certain subset $A'\subseteq A$)
should be taken advantage of. 
Some examples from real number computation 
show that such discrete advice can even make the
difference between computability and uncomputability.
We turn this into a both topological and combinatorial 
complexity theory of information,
investigating for several practical problems
how much advice is necessary and sufficient
to render them computable. 
\\
Specifically, 
finding a nontrivial solution to a homogeneous
linear equation $A\cdot\vec x=0$ for a given singular
real $n\times n$-matrix $A$ is possible when
knowing $\rank(A)\in\{0,1,\ldots,n-1\}$;
and we show this to be best possible.
Similarly, diagonalizing (i.e. finding a basis of eigenvectors of) 
a given real symmetric $n\times n$-matrix $A$ is possible 
when knowing the number of distinct eigenvalues: 
an integer between $1$ and $n$
(the latter corresponding to the nondegenerate case).
And again we show that $n$--fold
(i.e. roughly $\log n$ bits of) 
additional information is indeed necessary
in order to render this problem (continuous and) computable;
whereas for finding \emph{some single} eigenvector of $A$,
providing the truncated binary logarithm of the least-dimensional
eigenspace of $A$---i.e. $(\lfloor\log_2 n\rfloor+1)$-fold
advice---is sufficient and optimal.
\end{abstract}
\begin{minipage}[c]{0.97\textwidth}\vspace*{-8ex}%
\renewcommand{\contentsname}{}
\tableofcontents
\end{minipage}
\pagebreak\addtocounter{page}{-1}%
\section{Introduction} \label{s:Intro}
Recursive Analysis, that is Turing's \cite{Turing} theory of 
rational approximations with prescribable error bounds,
is generally considered a very realistic model 
of real number computation \cite{Braverman}.
Much research has been spent in `effectivizing' classical
mathematical theorems, that is replacing 
mere existence claims
\begin{enumerate}
\item[i)] ``{\it for all $x$, there exists some $y$ such that \ldots}''
\qquad with
\item[ii)] ``\it for all \emph{computable} $x$,
there exists some \emph{computable} $y$
such that \ldots''
\end{enumerate}
Cf. e.g. the \textsf{Intermediate Value Theorem}
in classical analysis \mycite{Theorem~6.3.8.1}{Weihrauch}
or the \textsf{Krein-Milman Theorem} from
convex geometry \cite{GeNerode}.
Note that Claim~ii) is non-uniform:
it asserts $y$ to \emph{be} computable
whenever $x$ is;
yet, there may be no way of \emph{converting}
a Turing machine $M$ computing $x$
into a machine $N$ computing $y$
\mycite{Section~9.6}{Weihrauch}.
In fact, computing a function $f:x\mapsto y$ 
is significantly limited by the 
sometimes so-called \textsf{Main Theorem},
requiring that any such $f$ be necessarily continuous:
because finite approximations to the argument $x$
do not allow to determine the value $f(x)$ 
up to absolute error smaller than the `gap'
$\limsup_{t\to x}f(t)-\liminf_{t\to x}f(t)$
in case $x$ is a point of discontinuity of $f$.
In particular any non-constant discrete-valued function on the reals is 
uncomputable---for information-theoretic 
(as opposed to recursion-theoretic) reasons.
Thus, Recursive Analysis is sometimes criticized
as a purely mathematical theory, rendering uncomputable
even functions as simple as Gau\ss' staircase \cite{Koepf}.

\subsection{Motivating Examples}
On the other hand many applications do provide,
in addition to approximations to the continuous argument $x$,
also certain promise or discrete `advice'; 
e.g. whether $x$ is integral or not.
And such additional information does render many otherwise
uncomputable problems computable:

\begin{myexample} \label{x:Staircase}
The Gau\ss{} staircase is discontinuous, hence uncomputable.
Restricted to integers, however, it is simply the
identity, thus computable.
And restricted to \emph{non-}integers, it is computable as well;
cf. \mycite{Exercise~4.3.2}{Weihrauch}.
Thus, one bit of additional advice (``\emph{integer or not}'')
suffices to make $\lfloor\,\cdot\,\rfloor:\IR\to\IZ$ computable.
\end{myexample}
Also many problems in analysis involving
compact (hence bounded) sets are discontinuous 
unless provided with some integer bound;
compare e.g. \mycite{Section~5.2}{Weihrauch}.
For a more involved illustration from computational linear algebra,
we report from \mycite{Section~3.5}{LA}
the following

\begin{myexample} \label{x:Diag}
Given a real symmetric $d\times d$ matrix $A$
(in form of approximations $A_n\in\IQ^{d\times d}$
with $|A-A_n|\leq2^{-n}$),
it is generally impossible, for lack of continuity
and even in the multivalued sense,
to compute (approximations to) any eigenvector of $A$.
\\
However when providing, in addition to $A$ itself,
the number of distinct eigenvalues
(i.e. \emph{not} counting multiplicities)
of $A$, finding the entire spectral resolution 
(i.e. an orthogonal basis of eigenvectors)
becomes computable.
\end{myexample}
Another case study on the benefit of additional discrete advice
to uniform computability is taken from \mycite{Lemma~2.8}{noCH}:
\begin{myexample} \label{x:Roux}
A closed subset $A\subseteq\IR^d$ is called $\psidg$--computable
if one can, given $\vec x\in\IR^d$, approximate the distance 
\begin{equation} \label{e:distA}
d_A(\vec x) \;= \;\min\big\{\|\vec x-\vec a\|_2 : \vec a\in A\big\} 
\end{equation}
from \emph{below}; more formally:
upon input of a sequence $\vec q_n\in\IQ^d$ with $\|\vec x-\vec q_n\|\leq2^{-n}$,
output a sequence $p_m\in\IQ$ with $\sup_m p_m=d_A(\vec x)$;
compare \mycite{Section~5.1}{Weihrauch}.
Similarly, $\psidl$--computability of $A$ means 
approximation of $d_A$ from \emph{above}.
\begin{enumerate}
\item[a)]
A finite set $A=\{\vec v_1,\ldots,\vec v_N\}\subseteq\IR^d$ is
$\psidl$--computable ~iff~ it is $\psidg$--computable ~iff~
each element $\vec v_i$ is computable.
\item[b)]
Neither of the three non-uniform equivalences in a) holds uniformly.
\item[c)]
However if the cardinality of $A$ is given as additional information,
$\psidl$--computability becomes uniformly equivalent to 
computability of $A$'s members
\item[d)]
whereas $\psidg$--computability still remains uniformly strictly weaker
than the other two.
\end{enumerate}
\end{myexample}
\begin{figure}[htb]
\centerline{\includegraphics[width=0.7\textwidth]{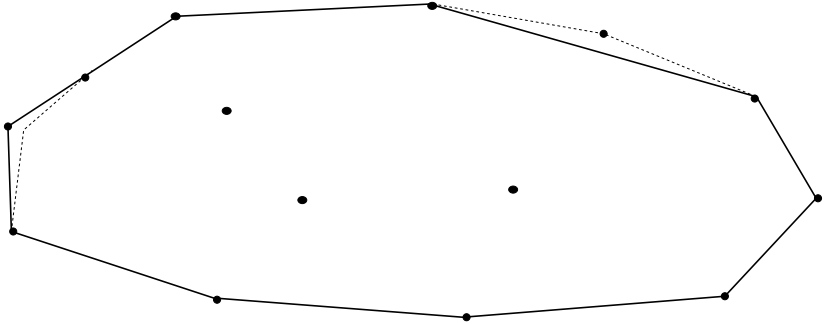}}%
\caption{\label{f:chull}The convex hull of some points in 2D.
Infinitesimal perturbation can heavily affect the (number and)
subset of extreme points.}
\end{figure}
Our next example treats a standard problem from computational geometry
\mycite{Section~1.1}{deBerg}:
\begin{myexample} \label{x:Chull}
For a set $S\subseteq\IR^d$,
its \emph{convex hull} is the least convex set containing $S$:
\[ \chull(S) \;:=\; \bigcap\big\{ C :
S\subseteq C\subseteq\IR^d, C\text{ convex}\big \} \enspace . \]
A \emph{polytope} is the convex hull of finitely many points,
$\chull(\{\vec p_1,\ldots,\vec p_N\})$. 
For a convex set $C$, point $\vec p\in C$ is called 
\emph{extreme} (written ``$\vec p\in\ext(C)$'')
if it does not lie on the interior of
any line segment contained in $C$:
\[ \vec p=\lambda\cdot\vec x+(1-\lambda)\cdot\vec y\;\wedge\;\vec x,\vec y\in C
\;\wedge\;0<\lambda<1 \quad\Rightarrow\quad \vec x=\vec y \enspace . \]
The problem 
\begin{equation} \label{e:Chull}
\extchull_N:
\binom{\IR^d}{N} \;\ni\; \{\vec x_1,\ldots,\vec x_N\}
\;\mapsto\; 
\big\{\vec y \text{ extreme point of }\chull(\vec x_1,\ldots,\vec x_N)\big\}
\end{equation}
of identifying the extreme points of the polytope $C$ spanned by 
given $\vec x_1,\ldots,\vec x_N$,
is discontinuous (and hence uncomputable) already in dimension $d=2$
and for $N=3$ with respect to output encoding $\psiG{}$,
cf. Figure~\ref{f:chull}:
\\
Let $\vec x_1:=(0,0)$, $\vec x_2:=(1,0)$, 
and $\vec x_3:=(\tfrac{1}{2},\epsilon)$:
For $\epsilon=0$, these points get mapped to $\{(0,0),(1,0)\}$;
whereas for $\epsilon\neq0$, the set of extreme points
is $\{(0,0),(1,0),(\tfrac{1}{2},\epsilon)\}$.
\end{myexample}
Trivially, $\extchull_N$ does become computable when giving,
in addition to approximations to the points $\vec x_1,\ldots,\vec x_N$,
one bit $b_i\in\{0,1\}$ for each $i=1,\ldots,N$
(that is, totally and in binary an integer between $0$ and $2^N-1$)
indicating whether $\vec x_i\in\ext\chull(\vec x_1,\ldots,\vec x_N)$.
However in Proposition~\ref{p:Chull} below we shall show
that, in order to compute $\extchull_N$, it suffices to 
know merely the \emph{number} $M\in\{2,\ldots,N\}$
of extreme points of $\chull(\vec x_1,\ldots,\vec x_N)$---and 
that $(N-1)$--fold discrete advice is in fact necessary.

\subsection{Complexity Measure of Non-Uniform Computability}
We are primarily interested in problems over real 
Euclidean spaces $\IR^d$, $d\in\IN$.
Yet for reasons of general applicability to arbitrary
spaces $U$ of continuum cardinality, we borrow from
Weihrauch's TTE framework \mycite{Section~3}{Weihrauch}
the concept of a so-called \emph{representation},
that is an encoding of all elements $u\in U$ as
infinite binary strings; and a \emph{realizer} of a
function $f:U\to V$ maps encodings of $u\in U$
to encodings of $f(u)\in V$.
A \emph{notation} is basically a representation
of a merely countable set.
Providing \emph{discrete advice} to $f$ amounts to presenting
to the Turing machine, in addition to an infinite binary string 
encoding $u\in U$, some integer (or `colour') $i$; and doing so 
for each $u$, means to color $U$. Now it is natural to 
wonder about the least advice (i.e. the minimum number of colors)
needed:

\begin{definition} \label{d:Nonunif}
\begin{enumerate}
\item[a)]
A function $f:\subseteq A\to B$ between topological spaces
$A$ and $B$ is \emph{$k$-continuous}
if there exists a covering 
(equivalently: a partition) $\Delta$ 
of $\dom(f)=\bigcup_{D\in\Delta} D$
with $\Card(\Delta)=k$
such that $f|_{D}$ is continuous for each $D\in\Delta$.
\\
Call\footnote{We are grateful for having been pointed out
that the \textsf{Continuum Hypothesis} is \emph{not} needed
in order to make this minimum well-defined. Anyway, in the following 
examples it will be at most countable, usually even finite.}
$\mycard(f):=\min\{k:f\text{ is $k$-continuous}\}$
the \emph{cardinal of discontinuity of $f$}.
\item[b)]
A function $f:\subseteq A\to B$ between represented spaces
$(A,\alpha)$ and $(B,\beta)$ is 
\emph{$(\alpha,\beta)$--computable with $k$-wise advice} 
(or simply \emph{$k$-computable} if $\alpha,\beta$ are clear from the context)
if there exists an at most countable partition $\Delta$
of $\Card(\Delta)=k$
and a notation $\delta$ of $\Delta$
such that the mapping
$f_{\Delta}:(a,D)\mapsto f(a)$ is $(\alpha,\delta,\beta)$--computable 
on $\dom(f_{\Delta}):=\{(a,D):a\in D\in\Delta\}$.
\\
Call $\mycomp(f)=\mycomp(f,\alpha,\beta):=
\min\{k:\text{$f$ is $(\alpha,\beta)$--computable with $k$-wise advise}\}$
the \emph{complexity of non-uniform $(\alpha,\beta)$--computability of $f$}.
\item[c)]
A function $f:\subseteq A\to B$ is
\emph{nonuniformly} $(\alpha,\beta)$--computable
if, for every $\alpha$--computable $a\in\dom(f)$, 
$f(a)$ is $\beta$-computable.
\end{enumerate}
\end{definition}
So continuous functions are exactly the
$1$-continuous ones; and computability
is equivalent to computability with $1$-wise advice.
Also we have, as an extension of the \textsf{Main Theorem} 
of Recursive Analysis, the following immediate

\begin{observation} \label{o:MainThm2}
If $\alpha,\beta$ are \emph{admissible} representations
in the sense of \mycite{Definition~3.2.7}{Weihrauch},
then every $k$-wise $(\alpha,\beta)$--computable function 
is $k$-continuous (but not vice versa);
that is $\mycard(f)\leq\mycomp(f)$ holds.\\
More precisely, every $k$-wise $(\alpha,\beta)$--computable 
possibly multivalued function 
$f:\subseteq A\toto B$ has a $k$-continuous
$(\alpha,\beta)$--realizer in the sense of
\mycite{Definition~3.1.3.4}{Weihrauch}.
\end{observation}
The above examples illustrate some interesting
discontinuous functions to be computable with
$k$-wise advice for some $k\in\IN$.
Specifically Example~\ref{x:Diag},
diagonalization of real symmetric $n\times n$--matrices
is $n$-computable; and Theorem~\ref{t:Diag} below
will show this value $n$ to be optimal.

\begin{myremark}
We advertise \textsf{Computability with Finite Advice} as
a generalization of classical Recursive Analysis:
\begin{enumerate}
\item[a)] It captures the concept of a hybrid approach 
  to discrete\&continuous computation.
\item[b)] It complements Type-2 oracle computation: \\
  In the discrete realm, \emph{every} function $f:\IN\to\IN$
  becomes computable
  when employing an appropriate oracle; whereas in the
  Type-2 case, exactly the \emph{continuous} functions $f:\IR\to\IR$
  are computable relative to some oracle
  \mycite{Corollary~6}{CIE}. On the other hand,
  2-wise advice can make a continuous function computable
  which without advice has unbounded degree of uncomputability;
  see Proposition~\ref{p:Comparison}d).
\item[c)]
  Discrete advice avoids a common major point of criticism 
  against Recursive Analysis, namely that it denounces
  even simplest discontinuous functions as uncomputable;
\item[d)]
  and such kind of advice is very practical: 
  In applications
  additional discrete information about the input 
  is often actually available and should be used.
  For instance a given real matrix may be known 
  to be non-degenerate (as is often exploited in numerics)
  or, slightly more generally, to have 
  $k$ eigenvalues coincide for some known $k\in\IN$.
\end{enumerate}
\end{myremark}

\subsection{Related Work, in particular Kolmogorov Complexity} \label{s:Related}
Definition~\ref{d:Nonunif} comes from \mycite{Definition~3.3}{Weihrauch92};
see also \mycite{Definition~5.8}{Pauly} where our quantity
$\mycard(f)$ it is called ``\emph{basesize}. 
Providing discrete advice can also be considered as yet another instance of 
\textsf{enrichment} in mathematics \cite[p.238/239]{KreiselMacIntyre}.

Various other approaches have been pursued in the literature in order to make
discontinuous functions accessible to nontrivial computability investigations.
\begin{description}
\item[\sf Exact Geometric Computation] 
  considers the arguments $\vec x$ as exact rational numbers \cite{EGC}.
\item[\sf Special encodings of discontinuous functions]
  motivated by spaces in Functional Analysis,
  are treated e.g. in \cite{Ning}; however these
  do not admit evaluation. 
\item[\sf Weakened notions of computability]
  may refer to stronger models of computation \cite{Hotz};
  provide more information on (e.g. the binary encoding of, rather 
  than rational approximations with error bounds to) the argument $x$ \cite{Mori,Fine};
  or expect less information on (e.g. no error bounds for approximations to)
  the value $f(x)$ \cite{SemiTCS}.
\item[\sf A taxonomy of discontinuous functions,]
  namely their \emph{degrees} of Borel measurability,
  is investigated in \cite{EffBorel,ToCS,CCA06}: \\
Specifically, a function $f:\subseteq A\to B$ is continuous (=$\Sigma_1$--measurable) ~iff,
for every closed $T\subseteq B$, its preimage $f^{-1}[T]$ is closed in $\dom(f)\subseteq A$;
and $f$ is computable iff this mapping $T\mapsto f^{-1}[T]$ on closed sets
is $(\psidg,\psidg)$--computable.
A degree relaxation, $f$ is called $\Sigma_2$--measurable ~iff,
for every closed $T\subseteq B$, $f^{-1}[T]$ is an $\text{F}_{\delta}$-set.
\item[\sf \emph{Wadge degrees} of discontinuity]
  are an (immense) refinement of the above,
  namely with respect to so-called \emph{Wadge reducibility};
  cf. e.g. \mycite{Section~8.2}{Weihrauch}.
\item[\sf \emph{Levels} of discontinuity]
  are studied in \cite{HertlingCCCG,HertlingDiss,Hertling}: \\
  Take the set $\LEV'(f,1)\subseteq\dom(f)=:\LEV'(f,0)$ 
  of points of discontinuity of $f=f|_{\LEV'(f,0)}$;
  then the set $\LEV'(f,2)\subseteq\LEV'(f,1)$ 
  of points of discontinuity of $f|_{\LEV'(f,1)}$
  and so on: the least index $k$ for which $\LEV'(f,k)=\emptyset$ holds
  is $f$'s level of discontinuity $\Lev'(f)$. \\
  A variant, $\Lev(f)$, considers $\LEV(f,1)$
  the closure of $\LEV'(f,1)$ in $\dom(f)$, 
  then $\LEV(f,2)$ the closure of points of discontinuity of $f|_{\LEV(f,1)}$
  and so on until $\LEV(f,k)=\emptyset$.
\end{description}
Our approach superficially resembles the third and last ones above.
A minor difference, they correspond to \emph{ordinal} measures
whereas the size of the partition considered in Definition~\ref{d:Nonunif}
is a \emph{cardinal}. 
As a major difference we now establish these measures 
as logically largely independent.
\begin{proposition} \label{p:Comparison}
\begin{enumerate}
\item[a)] 
There exists a $2$-computable function
$f:[0,1]\to\{0,1\}$ which is not measurable
nor on any level of discontinuity.
\item[b)]
There exists a $\Delta_2$--measurable function $f:[0,1]\to[0,1]$
with is not $k$-continuous for any finite $k$.
\item[c)]
If $f$ is on the $k$-th level of discontinuity,
it is $k$-continuous; in formula: 
$\mycard(f)\leq\Lev'(f)\leq\Lev(f)$.
\item[d)]
There exists a continuous, $2$-computable
function $f:\subseteq[0,1]\to[0,1]$
which is not computable,
even relative to any prescribed oracle.
\item[e)]
Every $k$-computable function is nonuniformly computable;
whereas there are nonuniformly computable functions
not $k$-computable for any $k\in\IN$.
\item[f)]
There even exists a nonuniformly computable $f:\IR\to\IR$
with $\mycard(f)=\frakc$, the cardinality of the continuum.
\end{enumerate}
\end{proposition}
Any real function is trivially $\frakc$-continuous by partitioning
its domain into singletons.
Item~f) is due to \person{Andrej Bauer}, personal communication.
Item~c) appears also in \mycite{Theorem~5.10}{Pauly}.
The last paragraph of \mycite{Section~5.1}{Pauly} includes
our Item~e) and partly extends Item~a) by exhibiting, 
to any ordinal $\lambda$ and cardinal $\beta\leq\lambda$, 
a function $f:\subseteq\IN^\lambda\to\beta$
with $\mycard(f)=\beta$ and $\Lev(f)=\lambda$.
Complementing Item~e),
conditions where nonuniform computability does imply
(even) $1$-computability 
have been devised in \cite{Invariance}.
\begin{proof}[Proposition~\ref{p:Comparison}]
\begin{enumerate}
\item[a)]
Consider a non Borel-measurable subset $S\subseteq[0,1]$;
e.g. exceeding the Borel hierarchy \cite{Hinman,Moschovakis}
by being complete for $\Delta_1^1$.
(Using the \textsf{Axiom of Choice}, $S$ can even
be chosen as non Lebesgue-measurable.)
Then its characteristic function $\cf{S}$ is not
measurable and totally discontinuous, hence
$[0,1]=\LEV'(\cf{S},1)=\LEV'(\cf{S},2)=\ldots$;
whereas $(S,[0,1]\setminus S)$ gives a
2-decomposition of $\dom(\cf{S})$ 
with $\cf{S}|_S\equiv 1$ and $\cf{S}|_{[0,1]\setminus S}\equiv0$.
\item[b)] See Example~\ref{x:Borel}b) below.
\item[c)]
By definition, $f$ is continuous
on $\LEV'(f,0)\setminus\LEV'(f,1)$, on $\LEV'(f,1)\setminus\LEV'(f,2)$, 
and so on---until $\LEV'(f,k-1)$ on which $f$ is continuous 
because $\LEV'(f,k)=\emptyset$.
Therefore $\Delta=\big(\LEV'(f,0)\setminus\LEV'(f,1),\LEV'(f,1)\setminus\LEV'(f,2),
\ldots,\LEV'(f,k-1)\big)$
constitutes a partition with the desired properties.
\item[d)]
Fix any uncomputable $t\in[0,1]$ and consider
\[ f:\subseteq[0,1]\to[0,1], \qquad
 f(x):=0\text{ for }x<t, \quad
 f(x):=1\text{ for }x>t, \quad f(t):=\bot \]
which is obviously continuous (because the
`jump' $x=t$ is not part of $\dom(f)$)
and 2-computable (namely on $[0,t)$ and $(t,1]$).
Since $t$ is uncomputable, $t\not\in\IQ$.
So if $f$ were computable, 
we could evaluate it at any $x\in\IQ$
to conclude whether $x<t$ or $x>t$;
and apply bisection to compute $t$ itself:
contradiction.
In fact we may choose $t$ uncomputable 
relative to any prescribed oracle
\cite{Xizhong,Barmpalias}.
\item[e)]
Let $f\big|_{D}$ be computable on each $D\in\Delta$.
Then $f(x)$ is computable for each computable $x\in D$;
hence also for each computable $x\in\dom(f)=\bigcup\Delta$.
\\
Example~\ref{x:Borel}b) below has range
$\{0\}\cup\{1/k:k\in\IN\}$ consisting of
computable (even rational) numbers only.
\item[f)]
Consider a \textsf{\Sierpinski-Zygmund Function} 
\mycite{Theorem~5.2}{Strange} $f:\IR\to\IR$,
i.e. such that $f|_D$ is discontinuous 
for any $D\subseteq\dom(f)$ of $\Card(D)=\frakc$.
Observe that this property is not affected by arbitrary modifications
of $f$ on any subset $X\subseteq\dom(f)$ of $\Card(X)<\frakc$:
If the restriction $f|_{\dom(f)\setminus X}$ is continuous on 
$D\setminus X$ for some $D\subseteq\dom(f)$ of $\Card(D)=\frakc$,
then so is $f$ on $D\setminus X$---contradicting $\Card(D\setminus X)=\frakc$.
\\
We may therefore modify the original function to be, say, identically 0 
on the countable subset $X:=\IRc$ of recursive reals,
thus rendering nonuniformly computable.
Now suppose $\Delta$ is any partition of $\IR$ of $\Card(\Delta)<\frakc$.
Then, by \mycite{Exercise~7.13}{CoriLascar},
\[ \frakc
\quad=\quad \Card(\IR)
\quad=\quad \sum\nolimits_{D\in\Delta}\Card(D)
\quad=\quad \max\big(\Card(\Delta),\sup\nolimits_{D\in\Delta}\Card(D)\big) \]
requires $\Card(D)=\frakc$ for some $D\in\Delta$;
but $f|_D$ is discontinuous, hence $\mycard(f)\geq\frakc$.
\qed\end{enumerate}\end{proof}
Further related research includes
\begin{description}
\item[\sf \emph{Computational} Complexity] of real functions;
  see e.g. \cite{Ko} and \mycite{Section~7}{Weihrauch}.
Note, however, that Definition~\ref{d:Nonunif} refers to a purely
\emph{information}-theoretic notion of complexity of a function
and is therefore more in the spirit of
\item[\sf Information-based Complexity] in the sense of \cite{Wozniakowski}.
 There, on the other hand, inputs are considered as real number entities
 given exactly; whereas we consider approximations to real inputs
 enhanced with discrete advice.
\item[\sf Finite Continuity] is being studied for \emph{Darboux Functions}
  in \cite{Pawlak,Marciniak}. It amounts to $d$-continuity
  for some $d\in\IN$ according to Definition~\ref{d:Nonunif}a).
\item[\sf Kolmogorov Complexity] has been investigated for finite strings
 and, asymptotically, for infinite ones; cf. e.g. \mycite{Section~2.5}{Vitanyi}
 and \cite{Staiger}.
 Also a kind of advice is part of that theory in form of
  \emph{conditional} complexity \mycite{Definition~2.1.2}{Vitanyi}.
\end{description}
We quote from \cite[\textsc{Exercise}~2.3.4abc]{Vitanyi}
the following

\begin{fact} \label{f:Kolmogorov}
An infinite string $\bar\sigma=(\sigma_n)_{n\in\omega}\in\Sigma^\omega$
is computable (e.g. printed onto a one-way output tape by some
so-called Type-2 or monotone machine; cf. {\rm\cite{Weihrauch,Schmidhuber}})
\begin{enumerate}
\item[a)]
iff~ its initial segments 
$\bar\sigma_{1:n}:=(\sigma_1,\ldots\sigma_n)$ have
Kolmogorov complexity $\leq\calO(1)$ \emph{conditionally to $n$},
i.e., iff $\Kol\big(\bar\sigma_{1:n}|n\big)$
is bounded by some $c=c(\bar\sigma)\in\IN$ independent of $n$.
\item[b)]
Equivalently: the \emph{uniform} complexity
$\Kolu(\bar\sigma_{1:n})
:=\Kol(\bar\sigma_{1:n};n)$ in the sense of
\mycite{Exercise~2.3.3}{Vitanyi}
is bounded by some $c$ for infinitely many $n$.
\end{enumerate}
\end{fact}
Recall that $\Kol(\bar\sigma_{1:n};n)$ is defined 
as the least size of a program computing any
(not necessarily proper) \emph{extension} of the function 
$\{1,\ldots,n\}\ni i\mapsto\sigma_i$
\mycite{Exercise~2.1.12}{Vitanyi};
i.e. in contrast to $\Kol(\bar\sigma_{1:n}|n)$,
only lower bounds $i$ to $n$ are provided.

\begin{proof}[Claim~b]
If $\bar\sigma$ is computable by some machine $M$,
then obviously a minor (and constant size) modification $M'$
of it will, given $n\in\IN$, print $\bar\sigma_{1:n}$.
Hence $\Kolu\big(\bar\sigma_{1:n}\big)\leq |\langle M'\rangle|$.
\\
Concerning the converse implication, 
observe that there are only $\calO(1)^c$ machines
of size $\leq c$. And for each of the infinitely many $n$,
at least one of them prints all initial segments of length
up to $n$. Hence by pigeonhole principle,
a single one of them does so for infinitely many $n$.
Which implies it does so even for all $n$.
\qed\end{proof}
\begin{definition} \label{d:Kolmogorov}
\begin{enumerate}
\item[a)]
For $\bar\sigma\in\Sigma^\omega$,
write $\Kol(\bar\sigma):=\sup_n \Kol\big(\bar\sigma_{1:n}|n\big)$
and $\Kol(\bar\sigma|\bar\tau):=\sup_n \Kol\big(\bar\sigma_{1:n}|n,\bar\tau\big)$,
where the Kolmogorov complexity conditional to an \emph{in}finite string
is defined literally as for a finite one
\mycite{Definition~2.1.1}{Vitanyi}.
\item[b)]
Similarly, let
$\Kolu(\bar\sigma|\bar\tau):=\sup_n \Kolu\big(\bar\sigma_{1:n}|\bar\tau\big)$.
\item[c)]
For a represented space $(A,\alpha)$ and $a\in A$,
write $\Kol(a):=\inf\{\Kol(\bar\sigma):\alpha(\bar\sigma)=a\}$
and $\Kolu(a):=\inf\{\Kolu(\bar\sigma):\alpha(\bar\sigma)=a\}$.
\end{enumerate}
\end{definition}
Note that we purposely do not consider some \emph{normalized} 
form like $\Kol(\bar\sigma_{1:n}|n)\pmb{/n}$
in order to establish the following

\begin{proposition} \label{p:Kolmogorov}
A function $F:\subseteq\Sigma^\omega\to\Sigma^\omega$
is computable with finite advice 
~iff~ 
the Kolmogorov complexity 
$\Kolu\big(F(\bar\sigma)|\bar\sigma\big)$ is 
bounded by some $c$
independent of $\bar\sigma\in\dom(F)$.
\end{proposition}
It seems that (at least the proof in \cite{Loveland} of)
Fact~\ref{f:Kolmogorov}a) is `too non-uniform'
for Proposition~\ref{p:Kolmogorov}
to hold with $\Kolu$ replaced by $\Kol$,
even for compact $\dom(F)$.
\begin{proof}
Suppose $\bar\sigma\mapsto F(\bar\sigma)$
is computable for $\bar\sigma\in D_i$
by Turing machine $M_i$. Then obviously
$\Kolu\big(F(\bar\sigma)|\bar\sigma\big)\leq
|\langle M_i\rangle|+|\bin(i)|$ is bounded
independent of $i\leq d$.
\\
Conversely consider,
as in the proof of Fact~\ref{f:Kolmogorov}b),
the $d\leq\calO(1)^c$ machines $M_i$ of size $\leq c$; 
and remember that, for each $\bar\sigma\in\dom(F)$ 
and given $\bar\sigma$, 
some $M_i$ outputs the entire (as opposed to just
some initial segments of the) infinite string 
$F(\bar\sigma)$. 
Let $D_i\subseteq\dom(F)$ denote the set of
those $\bar\sigma$ for which $M_i$ does so.
Then $M_i$ computes $F\big|_{D_i}$ and
$\dom(F)=\bigcup_{i=1}^d D_i$:
$F$ is computable with $d$--fold advice.
\qed\end{proof}

\section{Properties of the Complexity of Non-uniform Computability}
\begin{lemma} \label{l:Comp}
\begin{enumerate}
\item[a)]
Let $f:A\to B$ be $d$-continuous (computable)
and $A'\subseteq A$. Then the restriction $f|_{A'}$
is again $d$-continuous (computable).
\item[b)]
Let $f:A\to B$ be $d$-continuous (computable)
and $g:B\to C$ be $k$-continuous (computable).
Then $g\circ f:A\to C$ is $d\cdot k$-continuous (computable).
\item[c)]
If $f:A\to B$ is $(\alpha,\beta)$--computable with $d$-wise advice
and $\alpha'\reduceq\alpha$ and $\beta\reduceq\beta'$,
then $f$ is also $(\alpha',\beta')$--computable with $d$-wise advice.
\end{enumerate}
\end{lemma}
\begin{proof}
\begin{enumerate}
\item[a)]
Obviously, any partition $\Delta$ of $A$
induces one $\Delta':=\{D\cap A':D\in\Delta\}$ of $A'$
of at most the same cardinality.
\item[b)]
If $f$ is continuous (computable) on $A_i\subseteq A$
and $g$ is continuous (computable) on $B_j\subseteq B$,
then $g\circ f$ is continuous (computable)
on $A_i\cap f^{-1}[B_j]$:
$f$ is on any subset of $A_i$;
and so is $g$ on any subset of $B_j$,
particularly on the image of $A_i\cap f^{-1}[B_j]\subseteq B_j$ under $f$.
\item[c)]
obvious.
\qed\end{enumerate}\end{proof}
A minimum size partition $\Delta$ of $\dom(f)$ 
to make $f$ computable on each $D\in\Delta$ 
need not be unique: Alternative to
Example~\ref{x:Staircase}, we

\begin{myremark} \label{r:Delta}
Given a $\myrho$--name of $x\in\IR$ and 
indicating whether $\lfloor x\rfloor\in\IZ$
is even or odd suffices to compute $\lfloor x\rfloor$:
\\
Suppose $\lfloor x\rfloor=2k\in 2\IZ$
(the odd case proceeds analogously).
Then $x\in[2k,2k+1)$. Conversely,
$x\in[2k-1,2k+2)$, 
together with the promise
$\lfloor x\rfloor\in 2\IZ$, implies
$\lfloor x\rfloor=2k$.
Hence, given $(q_n)\in\IQ$ with $|x-q_n|\leq2^{-n}$,
$k:=2\cdot\big\lfloor q_1/2+\tfrac{1}{4}\big\rfloor$ 
(calculated in exact rational arithmetic)
will yield the answer.
\qed\end{myremark}

\subsection{Witness of $k$-Discontinuity} \label{s:Witness}
Recall that the partition $\Delta$ in Definition~\ref{d:Nonunif}
need not satisfy any (e.g. topological regularity) conditions.
The following notion turns out as useful in lower bounding the
cardinality of such a partition:

\begin{definition} \label{d:Flag}
\begin{enumerate}
\item[a)]
A \textsf{$d$-dimensional flag} $\calF$ in a topological Hausdorff space $X$
is a collection 
\[ x, \quad (x_n)_{_n}\;, \quad (x_{n,m})_{_{n,m}}\;, \quad (x_{n,m,\ell})_{_{n,m,\ell}}\;, 
\quad\ldots,\quad (x_{n_1,\ldots,n_d})_{_{n_1,\ldots,n_d}} \]
of a point and of (multi-)sequences\footnote{The generally more appropriate
concept is that of a \emph{Moore-Smith} sequence or \emph{net}.
However, being interested in second countable spaces,
we may and shall restrict to ordinary sequences.
Similarly, the Hausdorff condition is invoked for mere convenience.} 
in $X$ such that,
for each (possibly empty) multi-index $\bar n\in\IN^k$ ($0\leq k<d$),
it holds $x_{\bar n}=\lim\limits_{m\to\infty} x_{\bar n,m}$.
\item[b)]
$\calF$ is \textsf{uniform} if furthermore,
again for each $\bar n\in\IN^k$ ($0\leq k<d$) and for each $1\leq\ell\leq d-k$,
it holds
$x_{\bar n}=\lim\limits_{m\to\infty} 
x_{\bar n,\underbrace{\scriptscriptstyle m,\ldots,m}_{\ell\text{ times}}}$.
\item[c)]
For $f:\subseteq X\to Y$ and $x\in\dom(f)$
a \textsf{witness of discontinuity of $f$ at $x$}
is a sequence $x_n\in\dom(f)$ such that
$\lim_{n\to\infty} f(x_n)$ exists but differs from $f(x)$.
\item[d)]
For $f:\subseteq X\to Y$, 
a \textsf{witness of $d$-discontinuity of $f$}
is a uniform $d$-dimensional flag $\calF$ in $\dom(f)$ such that,
for each $k=0,1,\ldots,d-1$
and for each $\bar n\in\IN^k$
and for each $1\leq\ell\leq d-k$,
$\big(x_{\bar n,\underbrace{\scriptscriptstyle m,\ldots,m}_{\ell\text{ times}}}\big)_m$
is a witness of discontinuity of $f$ at $x_{\bar n}$.
\end{enumerate}
\end{definition}
Observe that, since $d$ is finite, 
we may always (although not effectively) 
proceed from a flag to a uniform one by iteratively
taking appropriate subsequences.
In fact, sub(multi)sequences of $d$-flags 
and of witnesses of discontinuity
are again $d$-flags and witnesses of discontinuity.

\begin{myexample} \label{x:Card}
Consider the mapping $\Card_d:\IR^n\ni(x_1,\ldots,x_d)\mapsto\Card\{x_1,\ldots,x_d\}$
and let $\vec X:=(0^d)$, 
$\vec X_n:=(1/n,0^{d-1})$,
$\vec X_{n,m}:=(1/n,1/n+1/m,0^{d-2})$,
$\vec X_{n_1,\ldots,n_k}:=(1/n_1,1/n_1+1/n_2,\ldots,1/n_1+\cdots+1/n_k,0^{d-k})$.
Then obviously $\lim_m \vec X_{n_1,\ldots,n_k,m,\ldots,m}=\vec X_{n_1,\ldots,n_k}$,
hence we have a uniform $d$-dimensional flag.
Moreover $\Card_d(\vec X_{n_1,\ldots,n_k})=k+1$ for $k=0,1,\ldots,d-1$.
shows it to be a witness of $(d-1)$-discontinuity.
\end{myexample}
Observe that $\Card_d$ is trivially $d$-continuous,
namely even constant on each 
$D_k:=\big\{(x_1,\ldots,x_d):\Card\{x_1,\ldots,x_d\}=k\big\}$,
$k=1,\ldots,d$. In fact $d$ is best possible as we have,
justifying the notion introduced in Definition~\ref{d:Flag}c),
the following

\begin{lemma} \label{l:Flag}
Let $X,Y$ be Hausdorff, $f:X\to Y$ a function,
and suppose there exists a witness 
of $d$-discontinuity of $f$.
Then $\mycard(f)>d$.
\end{lemma}
It also follows that 
Example~\ref{x:Roux}c) is best possible:
Knowing $k:=\Card\{x_1,\ldots,x_d\}\in\{1,\ldots,d\}$
(i.e. $d$-wise advice according to Example~\ref{x:Card})
is necessary for the computability of the members
\emph{without} repetition
of the (however) given set $\{x_1,\ldots,x_d\}$,
that is of a $\myrho^k$--name of 
$(x_{i_1},\ldots,x_{i_k})$ with 
$k=\Card\{x_{i_1},\ldots,x_{i_k}\}$.
Whereas computability of its members
\emph{with} repetition does not require
any advice according to \mycite{Lemma~5.1.10}{Weihrauch}, anyway.

\begin{proof}[Lemma~\ref{l:Flag}]
Suppose $\dom(f)=\biguplus_{i=1}^d D_i$ is a partition 
such that $f|_{D_i}$ is continuous; w.l.o.g. $x\in D_1$.
\\
Now consider the sequence $(x_n)$ in the flag:
$x_n\in\bigcup_{i=1}^d D_i$ implies
by pigeonhole that some $D_i$ contains infinitely many (w.l.o.g. all) $x_n$;
and $f(\lim_n x_n)=f(x)\neq\lim_{n\to\infty} f(x_n)$ requires $i\neq1$
in order for $f|_{D_i}$ to be continuous. W.l.o.g. $i=2$.
\\
We proceed to the double sequence $(x_{n,m})$ in the flag:
For each $n\in\IN$, some $D_{i(n)}\ni x_{n,m}$ for infinitely
many $m$; and $f(\lim_m x_{n,m})=f(x_n)\neq\lim_m f(x_{n,m})$
requires $i(n)\neq2$ for $f|_{D_{i(n)}}$ to be continuous.
Moreover some $i=i(n)$ for infinitely many $n$;
hence $f(\lim_m x_{m,m})=f(x)\neq\lim_m f(x_{m,m})$
also requires $i\neq1$. W.l.o.g. $i=3$.
\\
And so on until $i\neq1,2,3,\ldots,d$:
contradiction.
\qed\end{proof}

\subsection{Three Examples} \label{s:Examples}
Observe that for an $N\times M$-matrix $A$
and $d:=\min(N,M)$, $\rank(A)$ is an integer
between $0$ and $d$; and knowing this number
makes $\rank$ trivially computable.
Conversely, such $(d+1)$--fold information is
necessary, as follows from Lemma~\ref{l:Flag}
in connection with

\begin{myexample} \label{x:Rank}
Consider the space $\IR^{N\times M}$ of rectangular matrices
and let $d:=\min(N,M)$. For $i\in\{0,1,\ldots,d\}$ write
\begin{gather*}
E_i\quad:=\quad \sum_{j=1}^i 
\big( (0,\cdots,0,\underbrace{1}_{j\text{-th}},
0,\cdots,\underbrace{0}_{n\text{-th}})^\dagger
\otimes (0,\cdots,0,\underbrace{1}_{j\text{-th}},0,\cdots,\underbrace{0}_{m\text{-th}}) \big)
\quad=\\=\quad\left(\begin{array}{c@{\quad}c@{\quad}c@{\quad}c@{\quad}c@{\qquad}ccc}
1 & 0 & 0 &\cdots & 0 & 0 & \cdots & 0 \\
0 & 1 & 0 & \cdots & 0 & 0 & \cdots & 0 \\
0 & 0 & 1 &  & 0 & 0 & \cdots & 0 \\[-1ex]
\vdots & \vdots & & \ddots & 0 & 0 & \cdots & 0 \\
0 & 0 & 0 & \cdots & 1 & 0 & \cdots & 0 \\[1ex]
0 & 0 & 0 & \cdots & 0 & 0 & \cdots & 0 \\[-1ex]
\vdots & \vdots & \vdots & \vdots & \vdots & \vdots & \vdots & \vdots \\
0 & 0 & 0 & \cdots & 0 & 0 & \cdots & 0
\end{array}\right) \quad\in\quad \IR^{N\times M} \\
X:=0, \qquad X_{n_1,\ldots,n_i}\quad:=\quad E_1/n_1\;+\;E_2/n_2\;+\;\cdots\;+\;E_i/n_i
\end{gather*}
has 
$\lim_{m\to\infty} X_{n_1,\ldots,n_i,m,\ldots,m}=X_{n_1,\ldots,n_i}$,
hence constitutes a uniform $d$-dimensional flag.
Moreover, $\rank(E_i)=i=\rank(X_{n_1,\ldots,n_i})
\neq i+\ell=\rank(X_{n_1,\ldots,n_i,
\underbrace{\scriptstyle m,\ldots,m}_{\scriptscriptstyle \ell \text{ times}}})$
shows it to be a witness of $d$-discontinuity of $\rank:\IR^{N\times M}\to\{0,1,\ldots,d\}$.
\qed\end{myexample}
\begin{figure}[htb]
\centerline{%
\includegraphics[width=0.8\textwidth]{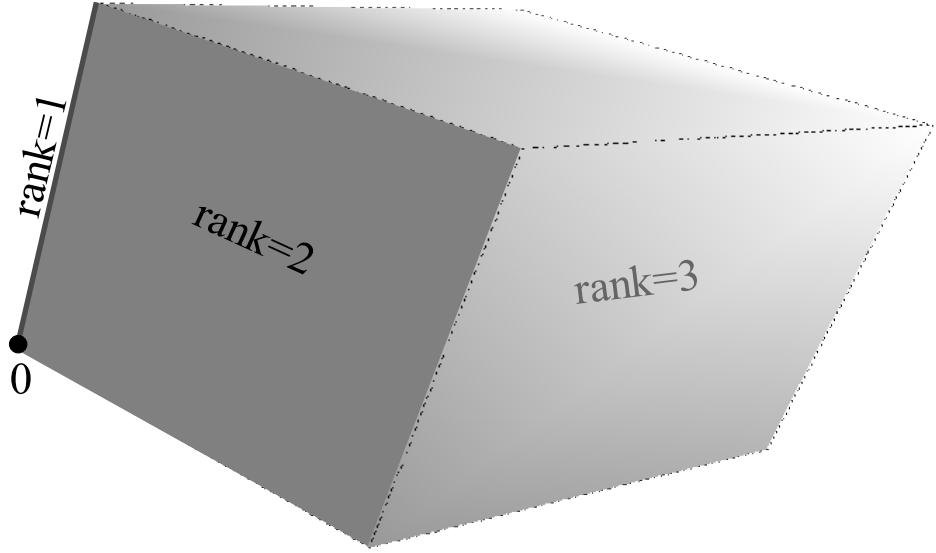}}%
\caption{\label{f:matrank}Schematic illustration of
the topology of the subspaces of rank-$k$ matrices, $k=0,1,2,3$.}
\end{figure}
\begin{myexample} \label{x:Borel}
Fix some bijection $\IN\times\IN\to\IN$,
$(x,y)\mapsto\langle x,y\rangle$;
e.g. $\langle x,y\rangle:=2^x\cdot(1+2y)$.
\begin{enumerate}
\item[a)]
 For $\bar n\in\IN^*$,
 let $\langle\bar n\rangle:=\sum_i 2^{-\langle i,n_i\rangle}$;
 and map the empty tuple to 0.\\
 This mapping $\langle\,\cdot\,\rangle:\IN^*\to[0,1]$ 
 is injective and maps to dyadic rationals.
 For each $k\in\IN$, the range $\langle\IN^k\rangle$ 
 belongs to $\Delta_2$; $\langle\IN^{\leq k}\rangle$ 
 is even closed a subset of $[0,1]$.
\item[b)]
 Consider $f:[0,1]\to[0,1]$ well-defined by
 $f(x):=1/k$ for $x=\langle\bar n\rangle$ with $\bar n\in\IN^k$;
 $f(x):=0$ for $x\not\in\langle\IN^*\rangle$. 
 Then $f$ is $\Delta_2$-measurable but not 
 $d$-continuous for any $d\in\IN$.
\end{enumerate} 
\end{myexample}
\begin{proof}[Example~\ref{x:Borel}]
\begin{enumerate}
\item[a)]
 Since the sum $\sum_{i\leq k}$ is finite for $\bar n\in\IN^k$,
 $\langle\bar n\rangle$ amounts to a dyadic rational,
 namely one with at most
 $k$ occurrences of the digit \texttt{1};
 the latter constitute a closed set.
\item[b)]
 Well-definition of $f$ follows from a).
 Moreover, $f^{-1}(1/k)=\langle\IN^k\rangle$ is in $\Delta_2$.
 Since $\range(f)=\{1/k:k\in\IN\}\cup\{0\}$,
 the preimage $f^{-1}[V]$ 
 of any open set $V\not\ni 0$
 is a union of finitely many $f^{-1}(1/k)$ 
 and therefore in $\Delta_2$, too;
 Whereas the preimage of open $V\ni 0$
 \emph{misses} finitely many $f^{-1}(1/k)$
 and thus also belongs to $\Delta_2$.\\
 Let $x:=0$, $x_n:=2^{-\langle 1,n\rangle}$, 
 $x_{n,m}:=2^{-\langle 1,n\rangle}+2^{-\langle 2,m\rangle}$,
 \ldots, $x_{n_1,\ldots,n_d}:=\sum_{i=1}^d 2^{-\langle i,n_i\rangle}$.
 This constitutes a uniform $d$-dimensional flag.
 And $f(x_{n_1,\ldots,n_k})=1/k\neq1/(k+\ell)=f(x_{n_1,\ldots,n_k,
 \underbrace{\scriptstyle m,\ldots,m}_{\scriptscriptstyle \ell \text{ times}}})$
 shows it to be a witness of $d$-discontinuity of $f$.
\qed\end{enumerate}\end{proof}
Recall Example~\ref{x:Chull} of computing (or rather identifying)
from a given $N$-tuple $(\vec x_1,\ldots,\vec x_N)$
of distinct points in $\IR^d$ those extremal to
(i.e. minimal and spanning) the convex hull
$\chull(\vec x_1,\ldots,\vec x_N)$.
In the 1D case, this problem $(x_1,\ldots,x_N)\mapsto\extchull_N(x_1,\ldots,x_N)$
is computable:
simply return the two (distinct!)
numbers $\min\{x_1,\ldots,x_N\}$ and $\max\{x_1,\ldots,x_N\}$.
We have already seen that in 2D it generally lacks
$\psiG{2}$--computability because of discontinuity.

\begin{proposition} \label{p:Chull}
Let $\vec x_1,\ldots,\vec x_N\in\IR^d$ be pairwise distinct
and $C:=\chull(\vec x_1,\ldots,\vec x_N)$.
\begin{enumerate}
\item[a)]
Let $\vec y\in\ext(C)$. Then there exists
a closed halfspace 
\[ H^+_{\vec u,t} \;=\; \big\{\vec z\in\IR^d: \sum\nolimits_i z_i u_i\geq t\big\} 
\;\subseteq\;\IR^d \]
with rational normal (although not necessarily unit) 
vector $\vec u\in\IQ^d\setminus\{0\}$ and $t>0$ such that
$H^+\cap\{\vec x_1,\ldots,\vec x_N\}=\{\vec y\}=\{\vec x_j\}$
for some $1\leq j\leq N$.
\item[b)]
Conversely $H^+_{\vec u,t}\cap\{\vec x_1,\ldots,\vec x_N\}=\{\vec x_j\}$ 
with $\vec u\neq0$ implies $\vec x_j\in\ext(C)$.
\item[c)]
Given $\vec x_1,\ldots,\vec x_N\in\IR^d$ as above
and for $1\leq j\leq N$, 
``$\vec x_j\in\ext(C)$'' is semi-decidable.
\item[d)]
The mapping $\extchull_N$
from Equation~\ref{e:Chull} is $(\myrho^{d\times N},\psidl)$--computable.
\item[e)]
For $N\geq2$ and given the number $M:=\Card\ext(C)\in\{2,\ldots,N\}$ of extreme points,
the set of their indices, i.e.
\[ \{i_1,\ldots,i_M\}\subseteq\{1,\ldots,N\} 
\quad\text{s.t.}\quad \ext(C)=\{\vec x_{i_1},\ldots,\vec x_{i_M}\} \]
becomes $(\myrho^{d\times N},\nu)$--computable.\\
In particular, $\extchull_N$
is $(\myrho^{d\times N},\psidg)$--computable with $(N-1)$-wise advice.
\item[f)]
It is however 
$(N-2)$-wise $(\myrho^{d\times N},\psidg)$--discontinuous
in dimensions $d\geq2$.
\end{enumerate}
\end{proposition}
\begin{proof}
\begin{enumerate}
\item[d)]
Follows from c) by trying all $j=1,\ldots,N$.
Indeed, a $\psidl$--name (but not a $\psidg$--name)
permits to `increase' at any time the set to be output.
\item[e)]
similarly to d), now trying all $M$-tuples $(i_1<i_2<\cdots<i_M)$ in $\{1,\ldots,N\}$.
Note that indeed $\Card\ext(C)\geq2$ because the $\vec x_i$ are pairwise distinct.
\item[c)]
Follows from a+b) by dovetailed search for some
$\vec u\in\IQ^d\setminus\{0\}$ with
$\langle\vec u,\vec x_j\rangle=:t>\langle\vec u,\vec x_i\rangle$
for all $i\neq j$,
where $\langle\vec u,\vec x\rangle:=u_1x_1+\cdots+u_dx_d$.
\item[a)] and b):
It is well-known \cite{Gruenbaum} that extreme points $\vec y$ 
of a polytope $C$ (although not necessarily of a general convex body)
are precisely its exposed points,
i.e. satisfy $\{\vec y\}=C\cap H_{\vec u,t}^+$ for some 
$t>0$ and $\vec u\in\IR\setminus\{0\}$. 
Equivalently: 
$\langle\vec u,\vec x_j\rangle>\langle\vec u,\vec x_i\rangle$
for all $i\neq j$---obviously a condition open in $\vec u$,
which therefore may be chosen from the dense subset
$\IQ^d\subseteq\IR^d$.
\item[f)]
We might construct a witness of $(N-2)$-discontinuity,
but take the more elegant approach of a reduction 
by virtue of Lemma~\ref{l:Comp}b).
To this end observe that semi-decidability
of \emph{in}equality makes
$\Card_n\IR^n\ni(x_1,\ldots,x_n)\to\Card\{x_1,\ldots,x_n\}$
$(\myrho^n,\myrhol)$--computable, i.e.
upper semi-continuous; hence by Example~\ref{x:Card},
$\Card_n$ must be $(n-1)$-wise \emph{lower semi}-discontinuous.
\\
Now let $x_1,\ldots,x_{n}\in\IR$ be given.
According to \mycite{Exercise~4.3.15}{Weihrauch}
suppose w.l.o.g. $x_1\geq x_2\geq \cdots \geq x_{n-1}\geq x_n=0$.
Then proceed to the following collection $X$ of 
$(n+1)$ points in 2D: $(0,0)$, $(1,x_1)$, $(2,x_1+x_2)$, \ldots,
$(n-1,x_1+\cdots+x_{n-1})$, $(n,x_n)$;
cf. Figure~\ref{f:chullcard}.
Let $f_n:\IR^n\to\IR^{2\times(n+1)}$
denote this computable mapping $(x_1,\ldots,x_n)\mapsto X$.
Observe that the sequence of slopes from points $\#i$ to $\#i+1$
is non-increasing because $x_i\geq x_{i+1}$;
and two successive slopes $(\#i-1\to\#i)$ and $(\#i\to\#i+1)$
coincide ~iff~ $x_i= x_{i+1}$;
which in turn is equivalent to point $\#i$ 
\emph{not} being extreme to $\chull(X)$.
In fact from a $\psiG{2}$--name of $\extchull(X)$
one can semi-decide $(i,x_1+\cdots+x_i)\not\in\extchull(X)$;
cf. e.g. \cite[\textsc{Lemma~25}c]{MLQ2}.
This yields a $(\psiG{2},\myrhog)$--computable mapping
$h_n:\extchull_{n+1}(X)\mapsto\Card\extchull_{n+1}(X)=\Card\{x_1,\ldots,x_n\}$ 
defined on the image of $\extchull_{n+1}\circ f_n$.
Now since $h_n\circ\extchull_{n+1}\circ f_n=\Card:\IR^n\to\{1,\ldots,n\}$ 
is $(n-1)$-wise lower semi-discontinuous by the above considerations,
Lemma~\ref{l:Comp}b) requires that 
$\extchull_{n+1}$ be $(n-1)$-wise $(\myrho^{n+1},\psiG{2})$--discontinuous.
\qed\end{enumerate}\end{proof}

\begin{figure}[htb]
\centerline{\includegraphics[width=0.7\textwidth]{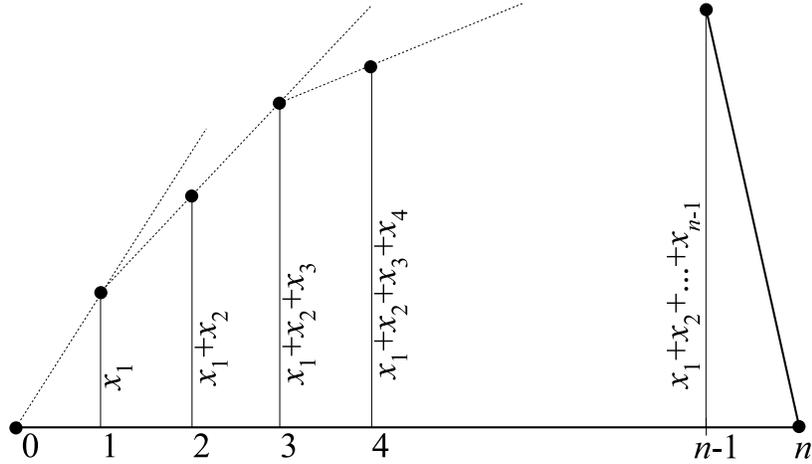}}%
\caption{\label{f:chullcard}Knowing in 2D which points are/not extreme to
their convex hull can be used to conclude which real numbers are
in-/equal: }
\end{figure}

\subsection{Further Remarks}
For some time the author had felt that 
when $\dom(f)$ is sufficiently `nice'
and for $x\in\dom(f)$,
the cardinal of discontinuity of $f$ could be lower bounded
in terms of the number of distinct limits
of $f$ at $x$, that is the cardinality of 
\[ \Lim(f,x) \quad:=\quad \big\{\lim\nolimits_{n\to\infty} f(x_n) : \dom(f)\ni x_n\to x\big\} 
\enspace . \]
However the following example 
(cf. also the right part of Figure~\ref{f:noLim3})
shows that this is not the case:
\[ f:[-1,1]\to[0,1], \qquad
2^{-n}\cdot 3^{-m}\mapsto 3^{-m} \quad (n,m\in\IN),
\qquad
f(x):\equiv0\text{ otherwise} \enspace .
\] 
Here $\Lim(f,0)$ is infinite but
$f$ is continuous on $D_1:=\{2^{-n}\cdot 3^{-m}:n,m\in\IN\}$
(because the latter set contains no accumulation point)
and $f\equiv 0$ on $D_2:=[-1,1]\setminus D_1$; hence $\mycard(f)=2$.
%
\begin{figure}[h]
\centerline{\includegraphics[width=0.45\textwidth]{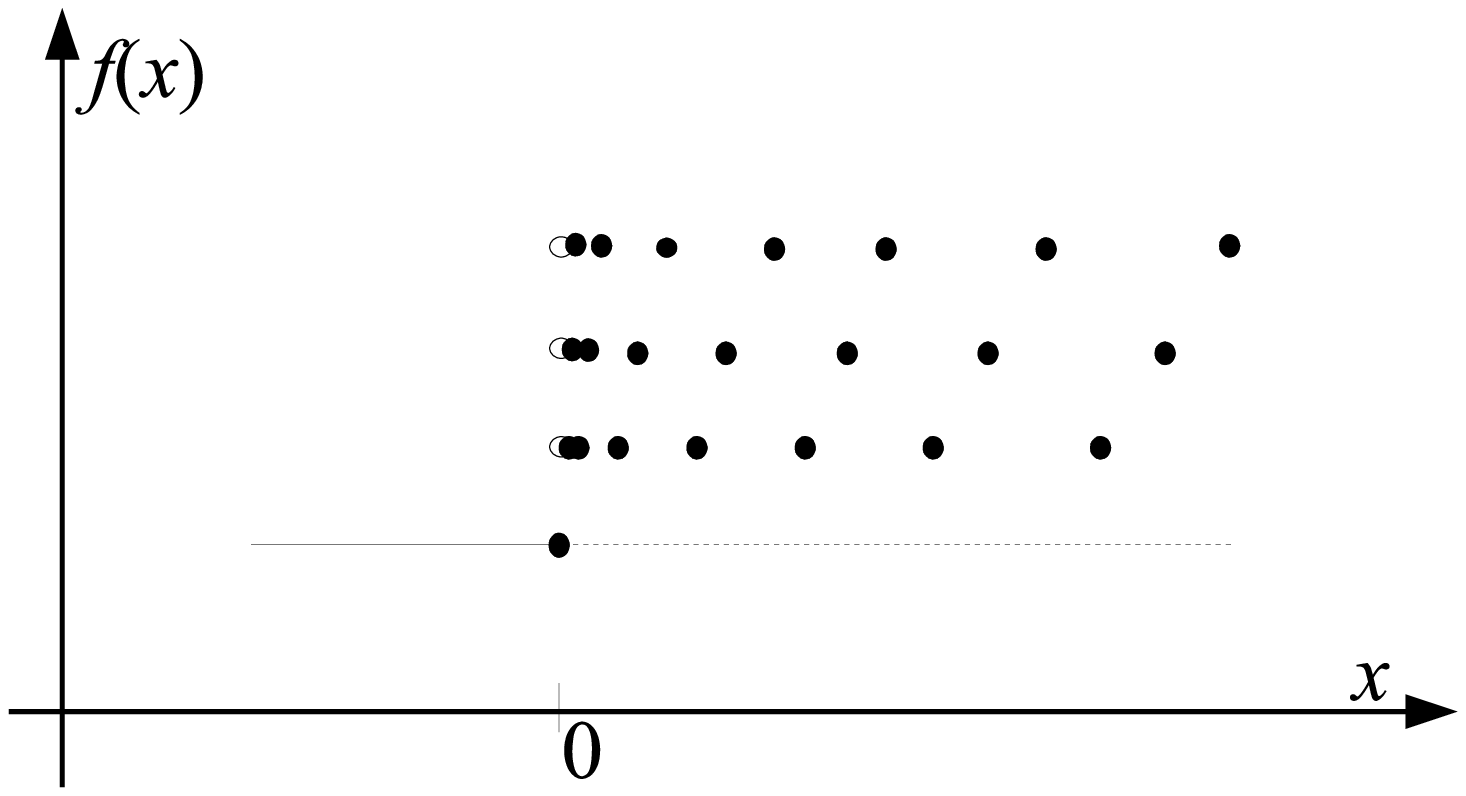}\qquad
\includegraphics[width=0.45\textwidth]{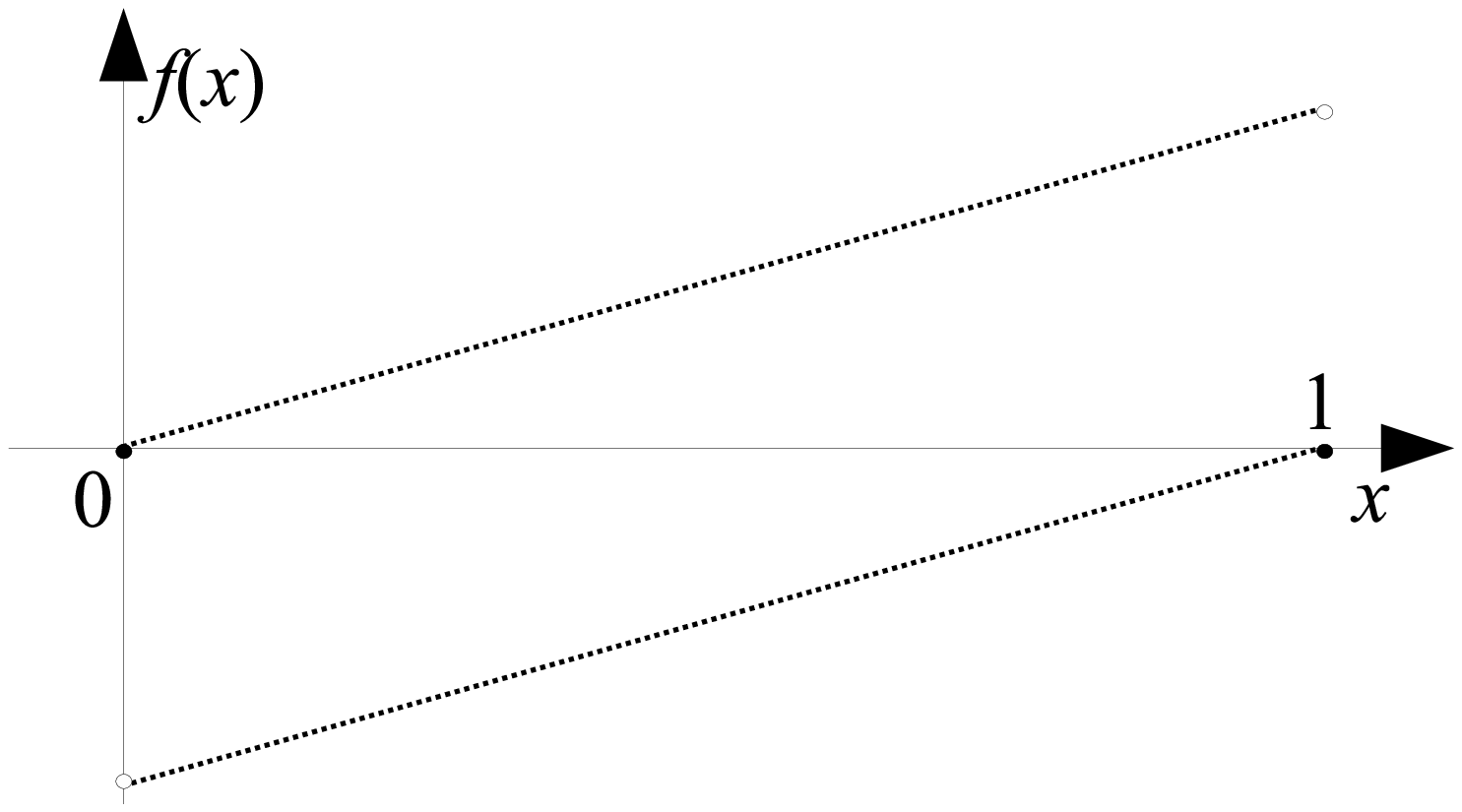}}
\caption{\label{f:noLim3}Left: the cardinal of discontinuity
cannot be lower bounded by the number of limit points.
Right: A 2-continuous function which,
after identifying arguments $x=0$ and $x=1$,
exhibits mere 3-continuity.}
\end{figure}

\smallskip\noindent
In order to apply Lemma~\ref{l:Flag} for
proving $k$-discontinuity of a function $f:A\to B$,
it may help to compactify the co-domain:

\begin{myexample}
Consider $f:[0,1]\to\IR$,
$0\mapsto 0$, $0<x\mapsto 1/x$. \\
Then $f$ admits no witness of 1-discontinuity; \\
whereas $\tilde f:[0,1]\to\IR\cup\{+\infty\}$
does admit such a witness.
\end{myexample}
The crucial point is of course that
$x:=0$ and $x_n:=1/n$ constitutes a witness of
1-discontinuity only for $\tilde f$,
because $0=f(x)\neq\lim_n f(x_n)=+\infty$
exists only in $\IR\cup\{+\infty\}$.
Finally we remark that
the notation $\delta$ in Definition~\ref{d:Nonunif}b)
is usually straight-forward and natural;
although an artificially bad choice is possible 
even for $2$-wise computable functions:

\begin{myexample}
The characteristic function $\chi_H:\IN\to\{0,1\}$
of the Halting problem $H\subseteq\IN$ is
obviously 2-wise $(\nu,\nu)$--computable
by virtue of $\Delta=\{H,\IN\setminus H\}$,
namely for $\delta:\subseteq\Sigma^*\to\Delta$ with
$1\mapsto H$ and $0\mapsto\IN\setminus H$.
\\
Whereas with respect to the following notation $\tilde\delta$,
$\chi_H$ is equally obviously 
\textsf{not} $(\nu,\tilde\delta,\nu)$--computable:
\[ \delta:\Sigma^*\to\Delta, \qquad
\bar x\mapsto H \text{ for } \bar x\in H, \quad
\bar x\mapsto\IN\setminus H \text{ for } \bar x\not\in H \enspace . \]
\end{myexample}
%
\subsection{Weak $k$-wise Advice} \label{s:Weak}
Recalling Observation~\ref{o:MainThm2},
(weak) $k$-wise $(\alpha,\beta)$-computability of $f:\subseteq A\to B$ implies
(weak) $k$-wise $(\alpha,\beta)$-continuity
from which in turn follows
\emph{weak} $k$-wise $(\alpha,\beta)$-continuity in the following sense:

\begin{definition} \label{d:Weak}
Consider a function $f:A\to B$ 
between represented spaces $(A,\alpha)$ and $(B,\beta)$.
\begin{enumerate}
\item[a)] 
Call $f$ \emph{$k$-wise $(\alpha,\beta)$-continuous}
if there exists a partition $\Delta$ of $\dom(f)$
of $\Card(\Delta)=k$ such that $f|_D$ is $(\alpha,\beta)$-continuous
on each $D\in\Delta$.
\item[b)]
Call $f$ \textsf{weakly} $k$-wise $(\alpha,\beta)$-continuous
if there exists a $k$-continuous $(\alpha,\beta)$-realizer
$F:\subseteq\Sigma^\omega\to\Sigma^\omega$ of $f$
in the sense of \mycite{Definition~3.1.3}{Weihrauch}.
\item[c)]
Call $f$ \textsf{weakly} $k$-wise $(\alpha,\beta)$-computable
if it admits a $k$-computable $(\alpha,\beta)$-realizer.
\end{enumerate}
\end{definition}
However conversely, and as opposed to the classical case $k=1$, 
weak $2$-wise $(\alpha,\beta)$--continuity in generally
does \textsf{not} imply $2$-wise $(\alpha,\beta)$-continuity.
Basically the reason is that a partition of $\dom(f)$ 
yields a partition of $\dom(F)$; whereas a partition
$\Delta$ of $\dom(F)$ need not be compatible with
the representation in that different $\alpha$ names
for the same argument $a$ may belong to different
elements of $\Delta$:
\begin{myexample} \label{x:Weak}
Consider the following function depicted 
to the right of Figure~\ref{f:noLim3}
\[ f:[0,1]\to[-1,+1],
\qquad [0,1)\cap\IQ\ni x\mapsto x=:g(x), 
\quad \IR\setminus\IQ\ni x\mapsto x-1=:h(x),
\quad f(1):=0. \]
It is continuous on both $\IQ\cap[0,1)$
and on $\{1\}\cup\IR\setminus\IQ$; hence 2-continuous,
and admits a 2-continuous $(\myrho,\myrho)$--realizer.

Now proceed from $[0,1]$ to $\calS^1$,
i.e. identify $x=0$ with $x=1$; 
formally, consider the representation 
$\tilde\myrho:=\imath\circ\myrho:\subseteq\Sigma^\omega\to\calS^1$
where $\imath:\IR\to[0,1)$, $x\mapsto x\mod 1$.
Since $f(0)=0=f(1)$, this induces a well-defined
function $\tilde f:\calS^1\to[-1,+1]$;
which admits a 2-continuous $(\tilde\myrho,\myrho)$--realizer:
namely the 2-continuous $(\myrho,\myrho)$--realizer of $f$.
But $\tilde f$ itself is not 2-continuous:
\\
Suppose $\calS^1=D_1\uplus D_2$ where $\tilde f|_{D_1}$
and $\tilde f|_{D_2}$ are both continuous.
W.l.o.g. $0\in D_1$. 
Observe that $\tilde f(0)=0=g(0)\neq h(0)=-1$.
Hence, as $\IQ$ is dense
and because continuous $h$ is different from continuous $g$,
continuity of $\tilde f|_{D_1}$ requires it to coincide with
$g$: first just locally at $x=0$, but then also globally---which
implies $\lim\limits_{x\nearrow 1} \tilde f|_{D_1}(x)=g(1)=1$,
contradicting $\tilde f|_{D_1}(1)=0$.
\qed\end{myexample}
As already mentioned, Example~\ref{x:Weak} illustrates
that the implication from $k$-wise $(\alpha,\beta)$-continuity
to weak $k$-wise $(\alpha,\beta)$-continuity cannot be reversed
in general---even for admissible representations.
Indeed, $\tilde\myrho$ can be shown equivalent
to the standard representation $\delta_{\calS^1}$
of $\calS^1$ as an effective topological space
\mycite{Definition~3.2.2}{Weihrauch}.

Applying Lemma~\ref{l:Comp} to realizers yields
the following counterpart for weak advice:

\begin{myremark} \label{r:Comp2}
Fix represented spaces $(A,\alpha)$, $(B,\beta)$,
and $(C,\gamma)$.
\begin{enumerate}
\item[a)]
Let $f:A\to B$ be weakly $d$-wise $(\alpha,\beta)$-continuous/computable
and $A'\subseteq A$. Then the restriction $f|_{A'}$
is again weakly $d$-wise $(\alpha,\beta)$-continuous/computable.
\item[b)]
Let $f:A\to B$ be weakly $d$-wise $(\alpha,\beta)$-continuous/computable
and $g:B\to C$ be weakly $k$-wise $(\beta,\gamma)$-continuous/computable.
Then $g\circ f:A\to C$ is weakly $d\cdot k$-wise $(\alpha,\gamma)$-continuous/computable.
\item[c)]
If $f:A\to B$ is weakly $d$-wise $(\alpha,\beta)$-continuous (computable)
and $\alpha'\reduceqt\alpha$ ($\alpha'\reduceq\alpha$)
and $\beta\reduceqt\beta'$ ($\beta\reduceq\beta')$,
then $f$ is also weakly $d$-wise $(\alpha',\beta')$--continuous (computable).
\end{enumerate}
\end{myremark}
Notice that property~b) does not carry over to \emph{multi-}representations
in the sense of \cite{Multirepresentations};
cf. the discussion preceeding Lemma~\ref{l:Comp2} below.

We also observe that Lemma~\ref{l:Flag} does not admit a converse,
even for total functions between compact spaces:

\begin{observation} \label{l:noFlag}
The function $\tilde f:\calS^1\to[-1,+1]$ from Example~\ref{x:Weak}
is not 2-continuous yet
has no witness of 2-discontinuity.
\end{observation}
\begin{proof}
Suppose $\big\{x,(x_n),(x_{n,m})\big\}$ is a
witness of 2-discontinuity of $\tilde f$.
First consider the
{\setlength{\parskip}{-9pt}\begin{itemize}
\item case $x\in(0,1)\cap\IQ$.
Since $x_n\to x$ and $x=\tilde f(x)\neq\lim_n\tilde f(x_n)$,
w.l.o.g. $0<x_n<1$ and $x_n\not\in\IQ$:
otherwise proceed to an appropriate subsequence.
Now $\lim_m x_{n,m}=x_n$ and 
$\lim_m \tilde f(x_{n,m})\neq\tilde f(x_n)=x_n-1$
requires, by definition of $\tilde f$,
$\tilde f(x_{n,m})=x_{n,m}$ for almost all $m$
and $n$: contradicting that a witness of discontinuity
is required to satisfy
$\lim_m \tilde f(x_{m,m})\neq\tilde f(x)$
and $\lim_m x_{m,m}=x$.
\item Case $x\in(0,1)\setminus\IQ$:
similarly.
\item Case $x=0\equiv 1$:
As $x_n\to x$ and since
$0\neq\tilde f(x)\neq\lim_n\tilde f(x_n)$ exists,
we may consider two subcases:
\item Subcase $x_n\in(1/2,1)\cap\IQ$
for almost all $n$: \\
Now $x_n=\lim_m x_{n,m}$ and 
$x_n=\tilde f(x_n)\neq\lim_m\tilde f(x_{n,m})$
requires, by definition of $\tilde f$, 
$\tilde f(x_{n,m})=x_{n,m}-1$ for almost all $m$
and $n$: contradicting $\lim_m x_{m,m}=x$ and
$\lim_m \tilde f(x_{m,m})\neq\tilde f(x)=0$.
\item Subcase $x_n\in(0,1/2)\setminus\IQ$
for almost all $n$: similarly.
\qed\end{itemize}}\end{proof}

\section{Multivalued Functions, i.e. Relations} \label{s:Multival}
Many applications involve functions
which are `non-deterministic' in the sense that,
for a given input argument $x$, several values $y$ 
are acceptable as output;
recall e.g. Items~i) and ii) in Section~\ref{s:Intro}.
Also in linear algebra,
given a singular matrix $A$, we want to find
\emph{some} (say normed) vector $\vec v$ such that
$A\cdot\vec v=0$. This is reflected by relaxing the
mapping $f:x\to y$ to be not a function but
a relation (also called multivalued function);
writing $f:X\toto Y$ instead of 
$f:X\to 2^Y\setminus\{\emptyset\}$
to indicate that for an input $x\in X$,
\emph{any} output $y\in f(x)$ is acceptable.
Many practical problems have been shown computable
as multivalued functions but admit no
computable single-valued so-called \emph{selection};
cf. e.g. \mycite{Exercise~5.1.13}{Weihrauch},
\cite[\textsc{Lemma~12} or \textsc{Proposition~17}]{LA}, 
and the left of Figure~\ref{f:multival} below.
On the other hand, even relations often lack computability 
merely for reasons of continuity---and 
appropriate additional discrete advice renders them
computable, recall Example~\ref{x:Diag} above.

Now Definition~\ref{d:Nonunif} of the 
complexity of non-uniform computability
straight-forwardly extends from single-valued
to multivalued functions; 
and Observation~\ref{o:MainThm2} relates them
to (single-valued) realizers;
which can then be treated using Lemma~\ref{l:Flag}.
However a direct generalization of Lemma~\ref{l:Flag}
to multivalued mappings turns out to be more convenient.
This approach requires a notion of (dis-)continuity
for relations rather than for functions.

\subsection{Continuity for Multivalued Mappings}
Like single-valued computable functions (recall the \textsf{Main Theorem}), 
also computable relations satisfy certain topological conditions.
However for such multivalued mappings, literature knows a variety
of easily confusable notions \cite{Neumaier}.
\emph{Hemicontinuity} for instance is not necessary
for real computability; cf. Example~\ref{x:Multival}a) below.
It may be tempting to regard \emph{computing} a multivalued mapping $f$ 
as the task of calculating, given $x$, the set-value $f(x)$ \cite{Spreen}.
In our example applications, however, one wants to capture that
a machine is permitted, given $x$, to `nondeterministically'
choose and output \emph{some} value $y\in f(x)$.
Note that this coincides with \mycite{Definition~3.1.3}{Weihrauch}.
In particular we do not insist that, upon input $x$,
\emph{all} $y\in f(x)$ occur as output for \emph{some}
nondeterministic choice---as required in \mycite{Section~7}{BrattkaCooper}.
Instead, let us generalize Definition~\ref{d:Flag} as follows:

\begin{definition} \label{d:Multival}
Fix some possibly multivalued mapping
$f:\subseteq X\toto Y$ 
and write $\dom(f):=\{x\in X:f(x)\neq\emptyset\}$.
\\
Call $f$ \textsf{continuous} at $x\in X$ if
there is some $y\in f(x)$ such that
for every open neighbourhood $V$ of $y$
there exists a neighbourhood $U$ of $x$
such that $f(z)\cap V\neq\emptyset$ for all $z\in U$.
\end{definition}
For ordinary (i.e. single-valued) functions $f$,
$\dom(f)$ amounts to the usual notion;
and such $f$ is obviously
continuous (at $x$) iff it is continuous (at $x$)
in the original sense.
\begin{figure}[htb]
\centerline{\includegraphics[width=0.95\textwidth]{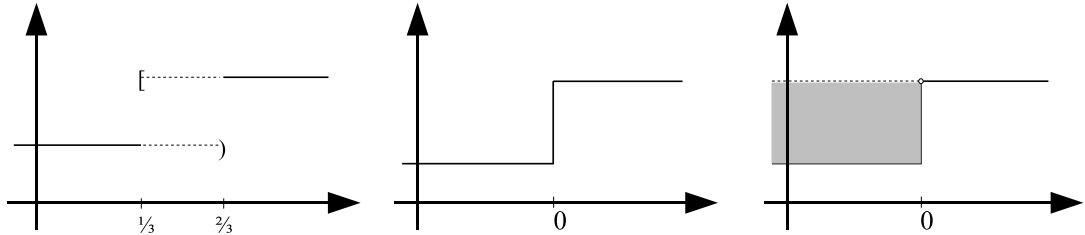}}
\caption{\label{f:multival}Left: A
$(\myrho,\myrho)$--computable relation which is not
hemicontinuous nor admits a
continuous selection. Middle: Quantification
over all $y\in f(x)$ is generally necessary
to capture discontinuity of a multivalued function.}
\end{figure}

\begin{myexample} \label{x:Multival}
\begin{enumerate}
\item[a)]
Consider the left of Figure~\ref{f:multival}, i.e. the multivalued function
\[ f:[0,1]\toto[0,1], \qquad
  1/3 > x\mapsto \{0\}, \quad
  [1/3,2/3)\ni x\mapsto \{0,1\}, \quad
  2/3 \leq x\mapsto \{1\}  \enspace . \]
Then $f$ is neither lower nor upper hemicontinuous---yet 
$(\myrho,\myrho)$--continuous, even computable:
Given $(q_n)\subseteq\IQ$ with $|x-q_n|\leq2^{-n}$,
test $q_3$: if $q_3\leq1/2$ output $0$, otherwise output $1$.
Indeed, $|x-q_3|\leq1/8$ implies $x\leq5/8<2/3$ for $q_3\leq1/2$,
hence $0\in f(x)$; whereas $q_3>1/2$ implies $x\geq3/8>1/3$,
hence $1\in f(x)$.
\item[b)]
Referring to the middle part of Figure~\ref{f:multival}, the multivalued function
\[ g:[-1,1]\toto[0,1], \qquad [-1,0)\ni x\mapsto \{0\},
\quad 0\mapsto[0,1], \quad (0,1]\ni x\mapsto\{1\} \]
is not continuous at 0 w.r.t. any $y\in f(0)=[0,1]$
\textsf{although} $f(0)$ itself 
does intersect $f(z)$ for all $z$.
\item[c)]
Consider the right part of Figure~\ref{f:multival}, i.e. the multivalued function
\[ h:[-1,+1]\toto[0,1], \qquad
0\geq x\mapsto [0,1), \quad 0<x\mapsto\{1\} \enspace . \]
Then $x_n:=2^{-n}$ constitutes a witness of discontinuity of $h$ at $x=0$
in the sense of Definition~\ref{d:Flag2}a) below:
For every $y\in h(x)=[0,1)$, $V:=(0,y/2+1/2)\ni y$
is an open neighbourhood of $y$ 
disjoint from $h(x_n)=\{1\}$ for all $n$.
\end{enumerate}
\end{myexample}
Lemma~\ref{l:Comp}a) literally applies also to multivalued mappings $f:A\toto B$.
Similarly generalizing Lemma~\ref{l:Comp}b) 
is quite cumbersome: For $B=\bigcup_i B_i$,
the preimages $f^{-1}[B_i]$,
{\setlength{\parskip}{-4pt}\begin{itemize}
\item if defined as $\{a\in A:f(a)\subseteq B_i\}$, need not cover $A$ 
\item if defined as $\{a\in A:f(a)\cap B_i\neq\emptyset\}$,
  need not be mapped to within $B_i$ by $f$.
\end{itemize}}\noindent
On the other hand, already the following partial generalization
of Lemma~\ref{l:Comp}b) turns out as useful:
\begin{lemma} \label{l:Comp2}
\begin{enumerate}
\item[a)]
Let $f:A\to B$ be single-valued and $g:B\toto C$ multivalued.
If $f$ is $d$-continuous (computable) 
and $g$ is $k$-continuous (computable),
then $g\circ f:A\toto C$ is $(d\cdot k)$-continuous (computable).
\item[b)]
Let $f:A\toto B$ and $g:B\toto C$ be multivalued.
If $f$ is $d$-continuous (computable)
and $g$ is continuous (computable),
then $g\circ f:A\toto C$ is again $d$-continuous (computable).
\end{enumerate}
\end{lemma}
\begin{proof}
\begin{enumerate}
\item[a)]
Since $f$ is single-valued, the set
$A_i\cap f^{-1}[B_j]$ is unambiguous and 
mapped by $f$ to a subset of $B_j$;
that is the proof of Lemma~\ref{l:Comp}b) carries over.
\item[b)]
If $f$ is continuous (computable) on each $A_i$,
then so is $g\circ f$.
\qed\end{enumerate}\end{proof}
Lemma~\ref{l:Flag2}a) below is an immediate extension of the
\textsf{Main Theorem of Recursive Analysis}, showing that
any computable \emph{multi}valued mapping is necessarily continuous.
It seems unknown whether also the converse, namely the
\textsf{Kreitz-Weihrauch Representation Theorem},
extends to the multivalued (for a start, real) case:

\begin{myquestion} \label{q:Kreitz}
Is the notion of multivalued continuity 
in Definition~\ref{d:Multival} strong enough
to assert that any function
$f:\subseteq\IR\toto\IR$ satisfying it
admits a Cantor-continuous 
$(\myrho,\myrho)$--realizer?
\end{myquestion}
%
\subsection{Witnesses of Discontinuity}
\begin{definition} \label{d:Flag2}
\begin{enumerate}
\item[a)]
For $x\in\dom(f)$,
a \textsf{witness of discontinuity of $f$ at $x$} 
is a sequence $(x_n)\in\dom(f)$ converging to $x$ such that,
for every $y\in f(x)$ there is some 
open neighbourhood $V$ of $y$
disjoint from $f(x_n)$ for
infinitely many $n\in\IN$.
\item[b)]
A uniform $d$-dimensional flag $\calF$ in $X$ 
is a \textsf{witness of $d$-discontinuity of $f$} if,
for each $0\leq k<d$
and for each $\bar n\in\IN^k$
and for each $1\leq\ell\leq d-k$ 
and for each $y\in f(x_{\bar n})$,
$\big(x_{\bar n,\underbrace{\scriptscriptstyle m,\ldots,m}_{\ell\text{ times}}}\big)_m$
is a witness of discontinuity of $f$ at $x_{\bar n}$.
\end{enumerate}
\end{definition}
If multivalued $f$ admits a witness of discontinuity at $x$,
then $f$ is not continuous.
Conversely, if $X$ is first-countable,
discontinuity of $f$ at $x$ yields the existence
of a witness of discontinuity at $x$. 
Also, witnesses of $1$-discontinuity 
coincide with witnesses of discontinuity;
and they generalize the definition from the single-valued case.
Lemma~\ref{l:Flag2} below extends Lemma~\ref{l:Flag}
in showing that a witness of $d$-discontinuity of $f$
inhibits $d$-computability.

\begin{lemma} \label{l:Flag2}
Let $(A,\alpha)$ and $(B,\beta)$ be effective metric spaces\footnote{%
Cf. \mycite{Section~8.1}{Weihrauch} for a formal definition
and imagine Euclidean spaces $\IR^k$ as major examples
and focus of interest for our purpose.}
with corresponding Cauchy representations
and $f:\subseteq A\toto B$ a possibly multivalued mapping.
\begin{enumerate}
\item[a)]
If $f$ admits a witness of discontinuity, 
then it is not $(\alpha,\beta)$--continuous.
\item[b)]
If $f$ admits a witness of $d$-discontinuity, 
then it is not $d$-wise $(\alpha,\beta)$--continuous.
\end{enumerate}
\end{lemma}
\begin{proof}
\begin{enumerate}
\item[a)]
Suppose $F:\subseteq\Sigma^\omega\to\Sigma^\omega$ 
is a continuous $(\alpha,\beta)$--realizer of $f$.
It maps some $\alpha$-name $\bar\sigma$
of $x$ to a $\beta$-name $\bar\tau$ of some $y\in f(x)$.
Now consider the neighbourhood $V\ni y$ according
to Definition~\ref{d:Multival}b).
By definition of the Cauchy representation $\beta$,
some finite initial part $(\tau_1,\ldots,\tau_M)=:\bar\tau|_{\leq M}$ 
of $\bar\tau$ restricts $y$ to belong to $V$;
and by continuity of $F$, this $\bar\tau|_{\leq M}$
depends on some finite initial part 
$\bar\sigma|_{\leq N}$ of $\bar\sigma$.
On the other hand $\bar\sigma|_{\leq N}$ is 
also initial part of an $\alpha$-name of some
element $x_n$ of the witness of discontinuity;
in fact of infinitely many of them.
But for $n$ sufficiently large,
$f(x_n)$ was supposed to not meet $V$; 
that is $\bar\tau|_{\leq M}$ is not initial part 
of a $\beta$-name of any $y'\in f(x_n)$:
contradiction.
\item[b)]
combines the arguments for a) with the proof of \ref{l:Flag}.
\qed\end{enumerate}\end{proof}
In comparison with the single-valued case,
a witness of discontinuity of a multivalued mapping 
involves one additional 
quantifier ranging universally over all $y\in f(x)$;
and Example~\ref{x:Multival}b) shows that this is
generally also necessary.
Nevertheless,
the following tool gives a (weaker yet) simpler condition
to be applied in Section~\ref{s:Applications}.

\begin{lemma} \label{l:Simpler}
Fix metric spaces $(A,\alpha)$ and $(B,\beta)$,
$\epsilon>0$, and $f:\subseteq A\toto B$.
\begin{enumerate}
\item[a)]
 For $S,T\subseteq A$,
 $\ball(S,\epsilon):=\{a\in A|\exists s\in S: d_A(a,s)<\epsilon\}$
 is disjoint from $T$ ~iff~ $\ball(T,\epsilon)$ is disjoint from $S$,
 and implies $\ball(S,\epsilon/2)\cap\ball(T,\epsilon/2)=\emptyset$.
\item[b)]
 Let $u_n$ and $v_n$ denote sequences in $\dom(f)$
 with $\lim_n u_n=x=\lim_n v_n$ such that
 $\ball\big(f(u_n),\epsilon\big)$ is disjoint from
  $f(v_n)$ for all but finitely many $n$.
 Then at least one of the sequences is a 
 witness of discontinuity of $f$ at $x$.
\item[c)]
 For $r\in\IN$ and $1\leq i\leq r$,
 let $(x_n^{(i)})_{_n}$ denote sequences in $\dom(f)$
 with $\lim_n x_n^{(i)}=x$ such that
 $\bigcap_{i=1}^r \ball\big(f(x_n^{(i)}),\epsilon\big)=\emptyset$
 holds for infinitely many $n$.
 Then, for some $i$, $(x_n^{(i)})_{_n}$ is a 
 witness of discontinuity of $f$ at $x$.
\item[d)]
 Fix $r,d\in\IN$ and
 consider a family of (multi-)sequences
\[ x, \quad x_{n_1}^{(i_1)}, \quad 
x_{n_1,n_2}^{(i_1,i_2)}, \quad \ldots, \quad
x_{n_1,\ldots,n_d}^{(i_1,\ldots,i_d)}, \qquad  n_1,\ldots,n_d\in\IN,
  \quad 1\leq i_1,\ldots,i_d\leq r \]
 such that, for each $\bar\imath\in\{1,\ldots,r\}^d$,
 $\big(x,x_{n_1}^{(i_1)},x_{n_1,n_2}^{(i_1,i_2)},\ldots,
 x_{n_1,\ldots,n_d}^{(i_1,\ldots,i_d)}\big)$ constitutes
 a uniform $d$-dimensional flag.
 Furthermore suppose that,
 for each $\bar n\in\IN^k$ and $\bar\imath\in\{1,\ldots,r\}^k$  ($0\leq k<d$) 
 and for each $1\leq\ell\leq d-k$,
\begin{equation} \label{e:Simpler}
\bigcap\nolimits_{j=1}^r \ball\Big(f\big(
x^{(\bar\imath,\overbrace{\scriptscriptstyle j,\ldots,j}^{\ell\text{ times}})}_{
  \bar n,\underbrace{\scriptscriptstyle m,\ldots,m}_{\ell\text{ times}}}\big),
\epsilon\Big)\quad=\quad\emptyset
\end{equation}
for infinitely many $m\in\IN$.
Then this family contains a witness of $d$-discontinuity of $f$.
\end{enumerate}
\end{lemma}
\begin{proof}
\begin{enumerate}
\item[a)]
 If $t\in T\cap\ball(S,\epsilon)$,
 there is some $s\in S$ with $d_A(s,t)<\epsilon$;
 hence $s\in S\cap\ball(T,\epsilon)$.
 So $T\cap\ball(S,\epsilon)\neq\emptyset$
 implies $S\cap\ball(T,\epsilon)\neq\emptyset$.
 The converse implication holds symmetrically. \\
 For $x\in\ball(S,\epsilon/2)\cap\ball(T,\epsilon/2)$
 there exist $s\in S$ and $t\in T$ with
 $d(s,x),d(t,x)<\epsilon/2$; hence 
 $d(s,t)<\epsilon$ by triangle inequality
 and $s\in S\cap\ball(T,\epsilon)$.
\item[b)]
 Suppose conversely that there exists some
 $y\in f(x)$ such that $V:=\ball(y,\epsilon/2)$
 intersects both $f(u_n)$ and $f(v_n)$
 for all $n\geq n_0$.
 Then $y\in\ball(\big(f(u_n),\epsilon/2\big)
 \cap\big(f(v_n),\epsilon\big)$;
 hence by a), $\ball\big(f(u_n),\epsilon\big)$
 intersects $f(v_n)$: contradiction.
\item[c)]
 similarly.
\item[d)]
 The case $r=1$ is that of c). We now treat $r=2$,
 the cases of higher values proceed similarly. \\
 By c) there is some $i$ such that 
 $\big(x_{n}^{(i)}\big)$ constitutes a
 witness of discontinuity of $f$ at $x$.
 Now consider the sequences
 $\big(x_{n,m}^{(i,j)}\big)_{_m}$ for
 $n\in\IN$ and $1\leq j\leq r$.
 Again by c), to each $n$ there is some $j(n)\in\{1,\ldots,r\}$
 such that $\big(x_{n,m}^{(i,j)}\big)_{_m}$
 is a witness of discontinuity of $f$ at $x_n^{(i)}$.
 According to pigeonhole, 
 $j(n)=j$ for some $j$ and for infinitely many $n$;
 hence we may proceed to an appropriate
 subsequence of $x_n$ and presume that 
 $\big(x_{n,m}^{(i,j)}\big)_{_m}$
 is a witness of discontinuity of $f$ at $x_n^{(i)}$
 for one common $j$;
 and $\big(x_{m,m}^{(i,j)}\big)_{_m}$
 a witness of discontinuity of $f$ at $x$:
 arriving at $\big(x,x^{(i)}_n,x^{(i,j)}_{n,m}\big)$
 a witness of 2-discontinuity of $f$.
\qed\end{enumerate}\end{proof}

\subsection{Example: Rational Approximations vs. Binary Expansion}
It is long known \cite{Turing2} that
a sequence of rational approximations 
to some $x\in\IR$ with error bounds
cannot continuously be converted into
a binary expansion of $x$. On the other hand
for \emph{non-}dyadic reals, i.e. for
\[ x \;\not\in\; \ID\;:=\;\big\{
(2r+1)/2^{k}: r,k\in\IZ\big\} \enspace , \]
such a conversion is computably possible
\mycite{Theorem~4.1.13.1}{Weihrauch};
while each \emph{rational} $x=r/s$ \underline{has} an
(ultimately periodic, hence) computable binary expansion
\mycite{Theorem~4.1.13.2}{Weihrauch}.
We observe that \emph{finding} such an expansion
for dyadic $x$ is infinitely discontinuous.
Recall that $\tfrac{1}{2}=(0.1000\ldots)_2=(0.01111\ldots)_2$
(and in fact each $x\in\ID$) admits two distinct
binary expansions.
\begin{proposition} \label{p:Rational}
The multivalued mapping
\[ \Adic_2:\ID\cap[0,1)\toto\{0,1\}^\omega, \quad
\sum_{m=1}^\infty b_m 2^{-m} \mapsto (b_1,b_2,\ldots) \]
is not $d$-wise $(\myrho,\nu^\omega)$-continuous for any
$d\in\IN$.
\end{proposition}
We remark that 
in fact for each $q=2,3,4,\ldots$,
the mapping $\Adic_q$ 
extracting $q$-adic expansions
is infinitely discontinuous on $\IQ$.
\begin{proof}[Proposition~\ref{p:Rational}]
Start with the rational sequence $q^{()}:=(\tfrac{1}{2},\tfrac{1}{2},\ldots)$,
a $\myrho$--name of $x:=\tfrac{1}{2}$.
Then consider the sequence of sequences 
\[ q^{(+)}_{n} \;\;:=\;\;\big(\underbrace{\tfrac{1}{2},\ldots,\tfrac{1}{2}}_{n\text{ times}},
\underbrace{\tfrac{1}{2}+2^{-n},\tfrac{1}{2}+2^{-n},\ldots}_{\infty\text{ times}}\big)
\enspace : \]
$\myrho$--names of $x^{(+)}_n:=\tfrac{1}{2}+2^{-n}$,
since $|x^{(+)}_n-q^{(0)}_{n,\ell}|=2^{-n}\leq 2^{-\ell}$ for $\ell\leq n$.
Moreover, $q^{(+)}_n$ converges in the Baire metric
to the sequence $q^{()}$. Similarly, 
$q^{(-)}_{n}:=(\tfrac{1}{2},\ldots,\tfrac{1}{2},
\tfrac{1}{2}-2^{-n},\tfrac{1}{2}-2^{-n},\ldots)$ is a $\myrho$--name 
of $x^{(-)}:=\tfrac{1}{2}-2^{-n}$ also converging to $q^{()}$.
And although the binary expansions of 
$x^{(+)}_n\neq x^{(-)}_n$ are both not unique,
\begin{align*}
x^{(+)}_n &\;=\;
(0.1\underbrace{0\cdots0}_{\makebox[0pt]{$\scriptscriptstyle(n-2)\text{ times}$}}100\ldots)_2 \;=\;
(0.1\underbrace{0\cdots0}_{\makebox[0pt]{$\scriptscriptstyle(n-2)\text{ times}$}}011\ldots)_2 \\
x^{(-)}_n &\;=\;
(0.0\underbrace{1\cdots1}_{\makebox[0pt]{$\scriptscriptstyle(n-2)\text{ times}$}}100\ldots)_2 \;=\;
(0.0\underbrace{1\cdots1}_{\makebox[0pt]{$\scriptscriptstyle(n-2)\text{ times}$}}011\ldots)_2
\end{align*}
shows that for $n\geq 2$
they must differ already in the first place.
Put differently, for 
$\epsilon:=1/2$ and with respect to Cantor metric,
$\ball\big(\Adic_2(x^{(+)}_n),\epsilon\big)\cap
\ball\big(\Adic_2(x^{(-)}_n),\epsilon\big)=\emptyset$:
a witness of discontinuity according to Lemma~\ref{l:Flag2}b).
\\
Next take
\[ q^{(\pm,\pm)}_{n,m} \;\;:=\;\; \big(
\underbrace{\tfrac{1}{2},\ldots,\tfrac{1}{2}}_{n\text{ times}},
\underbrace{\tfrac{1}{2}\pm2^{-n},\ldots,\tfrac{1}{2}\pm2^{-n}}_{m\text{ times}},
\underbrace{\tfrac{1}{2}\pm2^{-n}\pm2^{-n-m},1/2\pm2^{-n}\pm2^{-n-m},\ldots}_{\infty\text{ times}}
\big) \]
as $\myrho$--names of $x^{(\pm,\pm)}_{n,m}:=\tfrac{1}{2}\pm2^{-n}\pm2^{-n-m}
\not\in\{x,x^{(\pm)}_n:n\in\IN\}$
converging to $q^{(\pm)}_n$ for $m\to\infty$;
and $q^{(\pm,\pm)}_{m,m}$ to $q^{()}$.
Here, any binary expansions of $x^{(s,-)}$ and of
$x^{(s,+)}$ must differ in position $n$, i.e
$\ball\big(\Adic_2(x^{(s,+)}_{n,m}),\epsilon\big)\cap
\ball\big(\Adic_2(x^{(s,-)}_{n,m}),\epsilon\big)=\emptyset$
for $\epsilon\leq2^{-n}$; while still
$\ball\big(\Adic_2(x^{(+,s)}_{n,m}),\epsilon\big)\cap
\ball\big(\Adic_2(x^{(-,s)}_{n,m}),\epsilon\big)=\emptyset$
for $\epsilon\leq1/2$:
yielding a witness of 2-discontinuity
according to Lemma~\ref{l:Flag2}d)
with $r:=2$.
\\
And, continuing, $\displaystyle q^{(\pm,\ldots,\pm)}_{n_1,\ldots,n_d} :=$
\begin{align*}
\Big(&\underbrace{\tfrac{1}{2},\ldots,\tfrac{1}{2}}_{n_1\text{ times}},
\underbrace{\tfrac{1}{2}\pm2^{-n_1},\ldots,\tfrac{1}{2}\pm2^{-n_1}}_{n_2\text{ times}},
\underbrace{\tfrac{1}{2}\pm2^{-n_1}\pm2^{-n_1-n_2},\ldots,\tfrac{1}{2}\pm2^{-n_1}\pm2^{-n_1-n_2}}_{n_3\text{ times}},
\ldots,\\
&\ldots,\quad\ldots,\quad\ldots,\\
&\underbrace{\tfrac{1}{2}\pm2^{-n_1}\pm2^{-n_1-n_2}\pm\cdots\pm2^{-n_1-\cdots-n_{d-1}},
\ldots,\tfrac{1}{2}\pm2^{-n_1}\pm2^{-n_1-n_2}\pm\cdots\pm2^{-n_1-\cdots-n_{d-1}}}_{n_d\text{ times}},\\
&\underbrace{%
\tfrac{1}{2}\pm2^{-n_1}\pm2^{-n_1-n_2}\pm\cdots\pm2^{-n_1-\cdots-n_{d-1}}\pm2^{-n_1-\cdots-n_{d-1}-n_d},
\ldots}_{\infty\text{ times}}\Big)
\end{align*}
constitutes a compact $d$-flag of 
$\myrho$--names for $x^{(\pm,\ldots,\pm)}_{n_1,\ldots,n_d}:=
1/2\pm2^{-n_1}\pm2^{-n_1-n_2}\pm\cdots\pm2^{-n_1-\cdots-n_{d-1}}$
and witness of $d$-discontinuity of $\Adic_2$.
\qed\end{proof}
Now consider the problem
$\Adic_2^{(n)}:[0,1)\toto\{0,1\}^n$ of computing
only the first $n$ bits of the binary expansion
of $x$, given by rational approximations with error bounds.
Since $\Adic_2(x)$ is in a sense the limit of 
$\Adic_2^{(n)}(x)$ converging with $2^{-n}$ as $n\to\infty$,
it might seem natural to conjecture in view of Proposition~\ref{p:Rational}
that $\mycomp(\Adic_2^{(n)},\myrho,\nu)\to\infty$ for $n\to\infty$.
Indeed $2^n$-wise advice trivially suffices
for computing $\Adic_2^{(n)}(x)\in\{0,1\}^n$.
But one can do much better:

\begin{observation} \label{o:Rational}
$\Adic_2^{(n)}:[0,1)\toto\{0,1\}^n$ is 
$(\myrho,\nu)$--computable with 2-wise advice;
namely when giving, in addition to a $\myrho$--name
of $x$, also the $n$-th bit of its binary expansion.
\end{observation}
This can be considered an example of a phase transition.
(Note however that, implicitly, $n$ is given here.)

\begin{proof}[Observation~\ref{o:Rational}]
Suppose that $[0,1)\ni x=\sum_{i=1}^\infty b_i 2^{-i}$
with $b_i\in\{0,1\}$ and $b_n=0$.
(The other case $b_n=1$ proceeds analogously.)
Then it holds
\[ x \;\in\; \big[0,2^{-n}\big]\cup\big[2\cdot2^{-n},3\cdot2^{-n}\big]
\cup\cdots\cup\big[(2^n-2)\cdot2^{-n},(2^n-1)\cdot2^{-n}\big] \enspace , \]
corresponding to the $2^{n-1}$ possible choices of 
$(b_1,\ldots,b_{n-1},b_n)$ with $b_n=0$.
Conversely 
\begin{equation} \label{x:Rational}
x\in\big((2k-\tfrac{1}{2})\cdot2^{-n},(2k+\tfrac{3}{2})\cdot2^{-n}\big)
\quad\text{for}\quad k\in\{0,1,\ldots,2^{n-1}\}
\end{equation}
implies (since $b_n=0$)
$x\in\big[(2k)\cdot2^{-n},(2k+1)\cdot2^{-n}\big]$
and $(b_1,\ldots,b_{n-1})=\bin(k)$.
As \emph{strict} real inequalities are semi-decidable
(formally: $\myrho$--r.e. open in the sense of \mycite{Definition~3.1.3.2}{Weihrauch}),
dovetailing can search for $k$ to satisfy Equation~(\ref{x:Rational}).
\qed\end{proof}
\section{Applications} \label{s:Applications}
Based on Lemma~\ref{l:Flag}b), we now determine the
complexity of non-uniform computability for several
concrete functions 
including the examples from Section~\ref{s:Intro}.

\subsection[Linear Equation Solving: Linear Complexity]{Linear Equation Solving}
We first consider the problem of
solving a system of linear equations;
more precisely of finding a nonzero
vector in the kernel of a given singular matrix.
It is for mere notational convenience that we formulate
for the case of real matrices: complex ones work just as well.

\begin{theorem} \label{t:LinEq}
Fix $n,m\in\IN$, $d:=\min(n,m-1)$, and consider the space $\IR^{n\times m}$
of $n\times m$ matrices, considered as linear mappings
from $\IR^m$ to $\IR^n$. Then the multivalued mapping
\[ \LinEq_{n,m}:A\quad\mapsto\quad
\kernel(A)\setminus\{0\},\qquad
\dom(\LinEq)\;:=\;\{A\in\IR^{n\times m}:\rank(A)\leq d\}  \]
is well-defined and
has complexity $\mycard(\LinEq_{n,m})=\mycomp(\LinEq_{n,m},\myrho^{n\times m},\myrho^m)=d+1$.
\end{theorem}
\begin{proof}
Observe that $\{0\}\subsetneq\kernel(A)\subseteq\IR^m$
holds iff $\rank(A)\leq m-1$. Also $\rank(A)\leq n$ is
a tautology. Hence $\LinEq$ is totally defined.
\mycite{Theorem~11}{LA} has shown that
knowing $\rank(A)\in\{0,1,\ldots,d\}$ suffices for
computably finding a non-zero vector in
(and even an orthonormal basis of) $\kernel(A)$;
hence $\mycard(\LinEq)\leq\mycomp(\LinEq,\myrho^{n\times m},\myrho^m)\leq d+1$.

Conversely, we apply Lemma~\ref{l:Simpler}d) 
with $r:=m$ to assert $d$-discontinuity of $\LinEq$.
\\
Start with $A:=0^{n\times m}$, i.e. $\LinEq(A)=\IR^m\setminus\{0\}$.
Now Lemma~\ref{l:LinEq}a) below for $\delta:=1/N$ 
yields $m$ sequences $A^{(1)}_N,\ldots,A^{(m)}_N$ ($N\in\IN$)
with $\rank(A^{(i)}_N)\equiv 1$,
all converging to $A$
and with $\bigcap_i\kernel(A^{(i)}_N)=\{0\}$;
hence $\bigcap_i\LinEq(A^{(i)}_N)=\emptyset$.
However, Lemma~\ref{l:Simpler}c) requires
$\bigcap_i \ball\big(\LinEq(A^{(i)}_N),\epsilon\big)=\emptyset$
for some $\epsilon>0$.
On the other hand, observe that vector normalization 
\[ \norm : \IR^m\setminus\{0\} \;\ni\; \vec x\;\mapsto\; \vec x/\|\vec x\|\;\in\;\calS^{m-1} \]
is single-valued computable and continuous.
Hence by Lemma~\ref{l:Comp2}b) it suffices to prove
($d$-wise) discontinuity of $\LinEq':=\norm\circ\LinEq$.
Notice that $\LinEq'(A)=\LinEq(A)\cap\calS^{m-1}$ is compact. 
Thus, now, $\bigcap_i\LinEq(A^{(i)}_N)=\emptyset$ does imply 
\[ \bigcap\nolimits_i \ball\big(\LinEq'(A^{(i)}_N),\epsilon\big)
\quad\subseteq\quad
\ball\Big(\bigcap\nolimits_i \LinEq'(A^{(i)}_N),\delta\Big)
\quad=\quad\emptyset \]
for some appropriate $\epsilon>0$ according to 
(an inductive application of) Lemma~\ref{l:LinEq}b) below.
Indeed, $\epsilon$ can be chosen independent of $N$
since the subspaces $V_i$ from Lemma~\ref{l:LinEq}a)
do not depend on $\delta$. Hence we obtain by
Lemma~\ref{l:Simpler}c)---in a complicated way---a 
witness of (1-) discontinuity of $\LinEq'$.
\\
In case $d\leq 3$,
again applying Lemma~\ref{l:LinEq}a) similarly yields
rank-2 matrices $A^{(i,j)}_{N,M}$ ($j=1,\ldots,m-1$;
$A^{(i,m)}_{N,M}:=A^{(i)}_N$)
with $\lim_M A^{(i,j)}_{N,M}=A^i_N$
uniformly in $j,N$ and with
$\bigcap_j \kernel(A^{(i,j)}_{N,M})=\{0\}$,
hence again
$\bigcap\ball\big(\LinEq'(A^{(i)}),\epsilon\big)=\emptyset$
for some $\epsilon>0$ according to Lemma~\ref{l:LinEq}b):
and thus a witness of 2-discontinuity
by Lemma~\ref{l:Simpler}d).
\\
We may continue this process until arriving at
rank-$d$ matrices $A^{(i_1,\ldots,i_d)}_{N_1,\ldots,N_d}$
and a witness of $d$-discontinuity.
(And we cannot proceed any further because either
$d=\min(n,m)$ prohibits application of Lemma~\ref{l:LinEq}a
or, in case $d=m-1$, the matrices it yields exceed the
domain of $\LinEq$.)
\qed\end{proof}
The following tool, in addition to completing the proof of
Theorem~\ref{t:LinEq}, also gives further justification
for Figure~\ref{f:matrank}:
\begin{lemma} \label{l:LinEq}
\begin{enumerate}
\item[a)]
Let $n,m\in\IN$, $d:=\min(n,m)$, $A\in\IR^{n\times m}$,
$r:=\rank(A)<d$. There are subspaces
$V_1,\ldots,V_{m-r}$ of $\IR^m$ with $\bigcap_i V_i=\{0\}$
such that, to any $\delta>0$,
there exist $A^{(1)},\ldots,A^{(m-r)}\in\IR^{n\times m}$
with $\rank(A^{(i)})=r+1$, $\|A^{(i)}-A\|\leq\delta$ and
$\kernel(A^{(i)})=V_j$.
\item[b)]
Let $X,Y$ be closed subsets of $\IR^n$, $X$ compact, and $\delta>0$.
Then there exists $\epsilon>0$ such that
$\ball(X,\epsilon)\cap\ball(Y,\epsilon)\subseteq\ball(X\cap Y,\delta)$.
\end{enumerate}
\end{lemma}
\begin{proof}\begin{enumerate}
\item[a)]
Since $\rank(A)<n$, there exists some (w.l.o.g. normed)
$\vec w\in\IR^n\setminus\range(A)$.
Moreover by the \textsf{Rank-Nullity Theorem}, $\dim\kernel(A)=m-r$.
So consider an orthonormal
basis $\vec z_1,\ldots,\vec z_{m-r}\in\IR^m$ of $\kernel(A)$
and linear mappings 
\[
A^{(i)}:\IR^m\;\ni\;\vec x\;\mapsto 
\;A\cdot\vec x\;+\;\delta\cdot\langle\vec x,\vec z_i\rangle\vec w
\enspace : \]
These obviously satisfy $\|A^{(i)}-A\|\leq\delta$.
Moreover it holds $\range(A^{(i)})=\range(A)\oplus\lspan(\vec w)$
and $\kernel(A^{(i)})=\kernel(A)\cap\vec z_i^\bot=:V_i$:
because all $\vec z_j$ ($j\neq i$) are still mapped to $0$.
Hence $\rank(A^{(i)})=\rank(A)+1$ and 
$\bigcap_i V_i=\kernel(A)\bigcap_i\vec z_i^\bot=\{0\}$.
\item[b)] First consider the disjoint case $X\cap Y=\emptyset$.
Then the distance function $d_Y$ from Equation~(\ref{e:distA})
is positive on $X$. Moreover $d_Y$ is continuous
and therefore, on compact $X$, bounded from below by some $2\epsilon>0$.
Hence $\ball(X,\epsilon)\cap\ball(Y,\epsilon)=\emptyset$.
\\
In the general case, $Z:=X\cap Y$ is not necessarily empty but closed.
Now consider $X':=X\setminus\ball(Z,\delta/2)$
and $Y':=Y\setminus\ball(Z,\delta/2)$:
$X'$ is compact and disjoint from closed $Y'$;
hence $\ball(X',\epsilon)\cap\ball(Y',\epsilon)=\emptyset$
for some $0<\epsilon\leq\delta/2$
according to the first case.
Since $X\subseteq X'\cup\ball(Z,\delta/2)$,
\begin{multline*} \ball(X,\epsilon)\cap\ball(Y,\epsilon)
\;\subseteq\; \underbrace{\ball\big(X'\cup\ball(Z,\delta/2),\epsilon\big)}_{
\subseteq\ball(X',\epsilon)\cup\ball\big(\ball(Z,\delta/2),\epsilon\big)}
\;\cap\;\Big(\ball(Y',\epsilon)\cup
\underbrace{\ball\big(\ball(Z,\delta/2),\epsilon\big)}_{
=\ball(Z,\delta/2+\epsilon)\subseteq\ball(Z,\delta)}\Big)
\\\subseteq\;\big(\underbrace{\ball(X',\epsilon)\cap\ball(Y',\epsilon)}_{=\emptyset}\big)\;\cup\;
\big(\ball(X',\epsilon)\cap\ball(Z,\delta)\big)\cup
\big(\ball(Z,\delta)\cap\ball(Y,\epsilon)\big)
\quad\subseteq\quad\ball(Z,\delta)
\end{multline*}
\vspace*{-2.5ex}\qed\end{enumerate}\end{proof}
%
\subsection[Symmetric Matrix Diagonalization: 
Linear Complexity]{Symmetric Matrix Diagonalization} \label{s:Diag}
Similarly to Lemma~\ref{l:LinEq}a), we 
\begin{myremark} \label{r:Diag}
Let $\epsilon>0$ and 
let $A:\IC^n\to\IC^n$ denote an hermitian linear map
with $k$--fold degenerate eigenvalue $\lambda\in\IR$,
i.e. $\kernel(A-\lambda\id)=\lspan(\vec w)\oplus U$
for some eigenvector $\vec w$ (w.l.o.g. of norm 1)
orthogonal to a $(k-1)$--dimensional subspace $U\subseteq\IR^n$.
\\
Then the linear map $A':\IC^n\to\IC^n$ with 
$A'\big|_{\vec w^\bot}:\equiv A\big|_{\vec w^\bot}$
and $A':\vec w\mapsto (\lambda+\epsilon)\vec w$ is 
{\setlength{\parskip}{-9pt}\begin{itemize}
\item well-defined and hermitian (and real if $A$ was),
\item has $\|A-A'\|\leq\epsilon$ and
\item eigenspace to eigenvalue $\lambda$ cut down to $U$.
\vspace*{-1.5ex}\end{itemize}}\noindent%
Moreover, if $\epsilon$ is smaller than the difference between
 any two distinct eigenvalues of $A$, then
{\setlength{\parskip}{-9pt}\begin{itemize}
\item $\lambda+\epsilon$ is a new eigenvalue
\item with 1D eigenspace $\lspan(\vec w)$
\item while all other eigenspaces of $A'$ coincide with those of $A$.
\qed\end{itemize}}
\end{myremark}
\begin{figure}[htb]
\centerline{\includegraphics[width=0.95\textwidth]{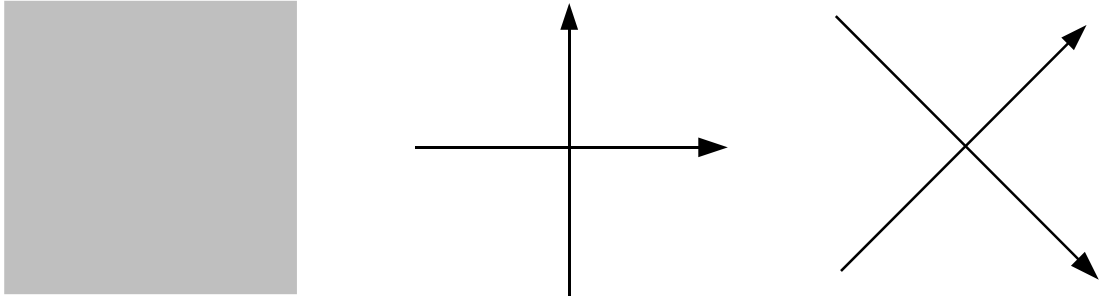}}%
\caption{\label{f:diag}Breaking 2--fold degeneracy of an eigenspace
to $A=\binom{0\;\;0}{0\;\;0}$
(left) in two ways 
$B_N=\binom{1\;\;0}{0\;\;1}/N$
and $C_N=\binom{0\;\;1}{1\;\;0}/N$
admitting no common eigenvectors (middle and right).}
\end{figure}
\begin{figure}[htb]
\centerline{\includegraphics[width=0.95\textwidth]{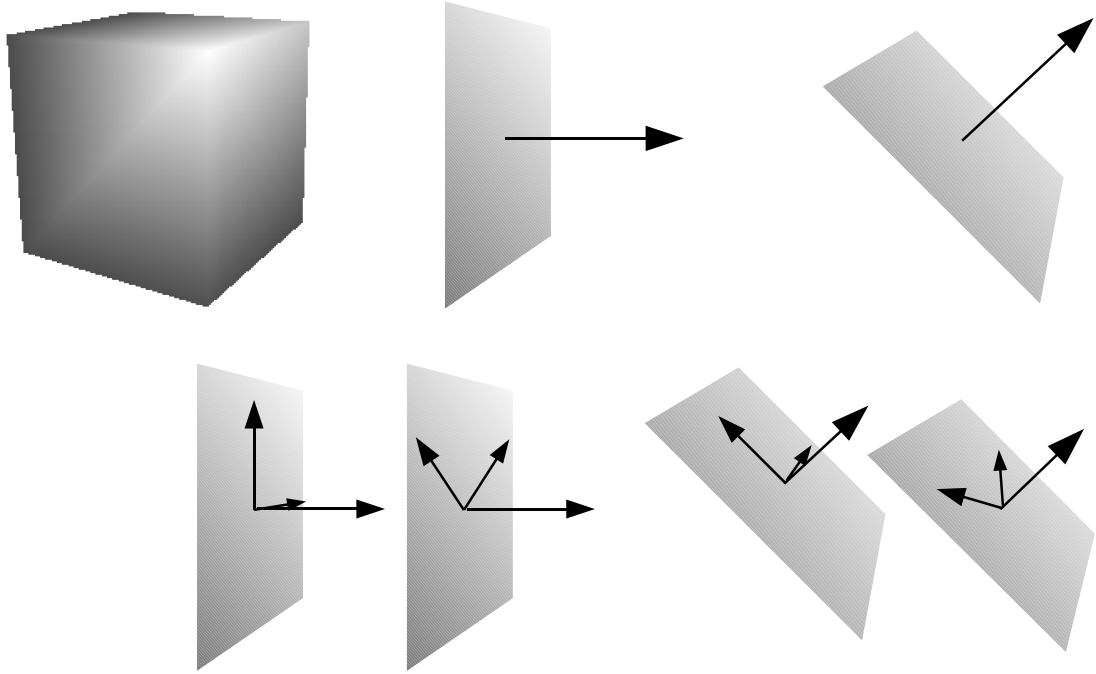}}%
\caption{\label{f:diag3}Construction similar to Figure~\ref{f:diag},
now iterated in 3D.}
\end{figure}
\begin{theorem} \label{t:Diag}
Fix $d\in\IN$ and consider the space $\IR^{\binom{d}{2}}$
of real symmetric $d\times d$ matrices.
Then the multivalued mapping
\[ \Diag_d:\IR^{\binom{d}{2}}\quad\ni\quad A\quad\mapsto\quad
\big\{(\vec w_1,\ldots,\vec w_d) \text{ basis of $\IR^d$ of eigenvectors
to $A$}\big\} \]
has complexity $\mycard(\Diag_d)=\mycomp\big(\Diag_d,\myrho^{\binom{d}{2}},\myrho^{d\times d}\big)=d$.
\end{theorem}
The lack of continuity of the mapping $\Diag$ is 
closely related to inputs with degenerate eigenvalues \mycite{Example~18}{LA}.
In fact our below proof yields a witness of $d$-discontinuity
by constructing an iterated sequence of \emph{symmetry breakings}
in the sense of Mathematical Physics; cf. Figure~\ref{f:diag3}.
On the other hand even in the non-degenerate case,
$\Diag$ is inherently multivalued since any permutation
of a basis constitutes again a basis.
\begin{proof}[Theorem~\ref{t:Diag}]
Let $\sigma(A)\subseteq\IR$ denote the set (!) of eigenvalues of $A$,
that is \emph{not} counting multiplicities.
\mycite{Theorem~19}{LA} has shown that knowing 
$\Card\sigma(A)\in\{1,\ldots,d\}$ suffices to 
compute some orthonormal basis of eigenvectors;
hence $\mycard(\Diag)\leq\mycomp\big(\Diag,\myrho^{\binom{d}{2}},\myrho^{d\times d}\big)\leq d$.

For the converse inequality, we shall apply Lemma~\ref{l:Simpler}d);
but, as in the proof of Theorem~\ref{t:LinEq}, 
first invoke Lemma~\ref{l:Comp2}b) by appending to $\Diag$
a computable mapping: namely orthonormalization. 
Indeed, standard \textsc{Gram-Schmidt} constitutes an effective procedure 
for turning a basis into an orthonormal one; and this process respects
eigenspaces because those belonging to different eigenvalues are
mutually orthogonal anyway. In the sequel we will therefore investigate
the multivalued mapping $\Diag':\IR^{\binom{d}{2}}\toto\calO(\IR,d)$
to the compact space of orthogonal matrices.
\\
Start with symmetric $A:=0^{d\times d}$, 
eigenvalue $0$ being $d$--fold degenerate
with eigenspace entire $\IR^d$, $d\geq2$.
Now consider two unit vectors $\vec v$ and $\vec w$ 
neither collinear nor orthogonal.
According to Remark~\ref{r:Diag} above
there exist corresponding 
sequences $(B_N)_{_N}$ and $(C_N)_{_N}$ 
of symmetric matrices,
both converging to $A$ and
with $(d-1)$--fold degenerate eigenspaces
and further 1D ones:
$B_N\cdot\vec v=1/N\cdot\vec v$ and
$C_N\cdot\vec w=1/N\cdot\vec w$; 
all eigenspaces are independent of $N$.
Hence, by Observation~\ref{o:Diag} below, 
$B_N$ and $C_N$ do not admit a common eigenvector basis;
i.e. $\Diag'(B_N)\cap\Diag'(C_N)=\emptyset$ for all $N$.
And Lemma~\ref{l:LinEq}b) implies 
$\ball\big(\Diag'(B_N),\epsilon\big)\cap
\ball\big(\Diag'(C_N),\epsilon\big)=\emptyset$
for some $\epsilon>0$ independent of $N$.
See also Figure~\ref{f:diag}\ldots
\\
We have satisfied the prerequisites to
Lemma~\ref{l:Simpler}b) and hence conclude
that there is a witness of discontinuity for $\Diag'$.
For a witness of 2-discontinuity 
observe that the above construction can be iterated
in case $d\geq3$ as depicted in Figure~\ref{f:diag3}:
To each $B_N=:A^{(0)}_N$ there exist sequences
$A^{(0,0)}_{N,M}$ and $A^{(0,1)}_{N,M}$ of symmetric
matrices converging to $A^{(0)}_N$ uniformly in $N$
with common $(d-2)$--fold degenerate eigenspaces
and further ones 
$A^{(0,j)}_{N,M}\cdot\vec x^{(0,j)}=1/M\cdot\vec x^{(0,j)}$
for $j=0,1$ where unit vectors $\vec x^{(0,0)}$ and $\vec x^{(0,1)}$
are neither collinear nor orthogonal;
similarly sequences $A^{(1,j)}_{N,M}$ and 
eigenvectors $\vec x^{(1,j)}$
correspond to $A^{(1)}_N:=C_N$.
Again, it follows 
$\ball\big(\Diag'(A^{(i,0)}_{N,M}),\epsilon\big)\cap
\ball\big(\Diag'(A^{(i,1)}_{N,M}),\epsilon\big)=\emptyset$
for some $\epsilon>0$ independent of $N,M$;
hence Lemma~\ref{l:Simpler}d) applies.
\\
And so on, until arriving at a witness of $(d-1)$-discontinuity
and at matrices $A^{(i_1,\ldots,i_{d-1})}_{N_1,\ldots,N_{d-1}}$
with $1$--fold (i.e. non-)degenerate eigenspaces
(where we cannot apply Remark~\ref{r:Diag} any more).
\qed\end{proof}
\begin{observation} \label{o:Diag}
Let $B,C\in\IC^{d\times d}$ be hermitian matrices.
Let $B\cdot\vec v=\lambda\vec v$ and $C\cdot\vec w=\mu\vec w$
denote respective eigenvectors to non-degenerate
eigenvalues $\lambda$ and $\mu$.
\\
If $B$ and $C$ admit a common eigenvector basis,
then $\vec v$ and $\vec w$ are either collinear or orthogonal.
\end{observation}
%
\subsection[Finding Some Eigenvector: 
Logarithmic Complexity]{Finding Some Eigenvector} \label{s:EV}
Instead of computing an entire basis of eigenvectors,
we now turn to the problem of determining just one
arbitrary eigenvector to a given real symmetric matrix.
This turns out to be considerably less `complex':

\begin{theorem} \label{t:EV1}
For a real symmetric $n\times n$-matrix $A$,
consider the quantity
\[
m(A)\quad:=\quad
\min\big\{\dim\kernel(A-\lambda\id):\lambda\in\sigma(A)\big\}
\quad\in\quad\{1,\ldots,n\} \enspace . \]
Given $d:=\lfloor\log_2 m\rfloor\in\{0,1,\ldots,\lfloor\log_2 n\rfloor\}$
and a $\myrho^{\binom{n}{2}}$--name of $A$, 
one can $\myrho^n$--compute (i.e. effectively find)
some eigenvector of $A$.
\end{theorem}
The proof employs the following tool about 
certain combinatorics and computability of finite multi-sets.
\begin{lemma} \label{l:Distinct}
Let $(x_1,\ldots,x_n)$ denote an $n$-tuple of real
numbers and consider the induced partition
$\calI:=\big\{\{1\leq i\leq n:x_i=x_j\}:1\leq j\leq n\big\}$
of the index set $\{1,\ldots,n\}=:[n]$
according to the equivalence relation
$i\equiv j:\Leftrightarrow x_i=x_j$.
Furthermore let $m:=\min\big\{\Card(I):I\in\calI\big\}$.
\begin{enumerate}
\item[a)]
 Consider $I\subseteq[n]$ with $1\leq\Card(I)<2m$ such that
\begin{equation} \label{e:Distinct}
x_i\neq x_j \quad\text{for all}\quad i\in I
\quad\text{and all}\quad j\in[n]\setminus I \enspace . 
\end{equation}
 Then $I\in\calI$.
\item[b)]
 Suppose $k\in\IN$ is such that $k\leq m<2k$.
 Then there exists $I\in\calI$ with $k\leq\Card(I)<2k$
 satisfying (\ref{e:Distinct}).
 Conversely every $I\subseteq[n]$ with $k\leq\Card(I)<2k$
 satisfying (\ref{e:Distinct}) has $I\in\calI$.
\item[c)]
 Given a $\myrho^n$--name of $(x_1,\ldots,x_n)$
 and given $k\in\IN$ with $k\leq m<2k$,
 one can computably find some $I\in\calI$.
\item[d)]
 Given a $\myrho^n$--name of $(x_1,\ldots,x_n)$
 and given $\Card(\calI)$, one can compute $\calI$.
\end{enumerate}
\end{lemma}
Claim~c) can be considered a weakening of Claim~d) 
which had been established in \mycite{Proposition~20}{LA}.
\begin{proof}
\begin{enumerate}
\item[a)]
 Take $i\in I$, $i\in J$ for some $J\in\calI$.
 Then obviously $I\supseteq J$, because $j\in J\setminus I$
 would imply $x_i=x_j$: contradicting Equation~(\ref{e:Distinct}).
\\
 It remains to show $I\subseteq J$.
 Suppose that $x_i\neq x_{i'}$ for some $i'\in I$.
 Then $i'\in J'$ for some $J'\in\calI$ disjoint to $J$.
 Thus condition ``$x_i\neq x_j$'' fails for all $j\in J$;
 and ``$x_{i'}\neq x_j$'' fails for all $j\in J'$:
 i.e. for a total of $\Card(J)+\Card(J')\geq 2m$ choices of $j\in[n]$,
 whereas by Equation~(\ref{e:Distinct})
 it is supposed to hold for all $j\in[n]\setminus I$:
 a total of $>n-2m$ choices---contradiction.
\item[b)]
 For the first claim, simply choose $I\in\calI$
 with $\Card(I)=m$. Concerning the second claim,
 observe that $k\leq m$ and $\Card(I)<2k$ imply
 $\Card(I)<2m$; hence Item~a) applies.
\item[c)]
 Recall that \emph{in}equality of real numbers is 
 `semi-decidable'; formally:
 $\{(x,y):x\neq y\}\subseteq\IR^2$ is
  $\myrho^2$--r.e. open in $\IR^2$
 in the sense of \mycite{Definition~3.1.3.2}{Weihrauch}.
 Hence we may simultaneously try every 
 $I\subseteq[n]$ with $k\leq\Card(I)<2k$
 and semi-decide Condition~(\ref{e:Distinct}):
 according to Item~b) this will succeed
 precisely for some $I\in\calI$.
\qed\end{enumerate}\end{proof}

\begin{proof}[Theorem~\ref{t:EV1}]
Compute according to \mycite{Proposition~17}{LA}
some ($\myrho^n$--name of an)  $n$-tuple of 
eigenvalues $(\lambda_1,\ldots,\lambda_n)$
of $A$, repeated according to their multiplicities.
Now due to \mycite{Theorem~11}{LA},
(some eigenvector in) the eigenspace $\kernel(A-\lambda_i\id)$
can be computably found when knowing $\rank(A-\lambda_i\id)$
(recall Theorem~\ref{t:LinEq}), that is the
multiplicity of $\lambda_i$ in the multi-set
$(\lambda_1,\ldots,\lambda_n)$.
To this end we apply Lemma~\ref{l:Distinct}c),
observing $k:=2^d\leq m<2k$ since $d=\lfloor\log_2 m\rfloor$.
\qed\end{proof}
\begin{theorem} \label{t:EV}
The multivalued mapping
\[ \EVec_n:\IR^{\binom{n}{2}}\quad\ni\quad A\quad\mapsto\quad
\{\vec w \text{ eigenvector of }A\} \]
has complexity 
$\mycard(\EVec_n)=\mycomp\big(\EVec_n,\myrho^{\binom{n}{2}},\myrho^{n}\big)=
\lfloor\log_2 n\rfloor+1$.
\end{theorem}

\begin{myremark} \label{r:EV}
\begin{enumerate}
\item[a)]
In $\IK^{2m}$, the two subspaces
$U:=\{ (x_1,\ldots,x_m,0,\ldots,0) :x_i\in\IK\}$ and
$V:=\{ (x_1,\ldots,x_m,x_1,\ldots,x_m):x_i\in\IK\}$,
as well as their orthogonal complements
$U^\bot=\{(0,\ldots,0,x_1,\ldots,x_m) :x_i\in\IK\}$ and
$V^\bot=\{ (x_1,\ldots,x_m,-x_1,\ldots,-x_m):x_i\in\IK\}$,
have dimension $m$ and satisfy
$\{0\}=U\cap V=U\cap V^\bot=U^\bot\cap V=U^\bot\cap V^\bot$.
\item[b)]
Write $W^{(0)}_0:=U$, $W^{(0)}_1:=U^\bot$,
$W^{(1)}_0:=V$, and $W^{(1)}_1:=V^\bot$.
Applying the above construction to
each of them again, we iteratively obtain
in $\IK^{2^d}$ subspaces $W^{(j_1,\ldots,j_k)}_{i_1,\ldots,i_k}$
of dimension $2^{d-k}$ ($0\leq k\leq d$, $i_\ell,j_\ell=0,1$) 
with the following properties:
{\setlength{\parskip}{-9pt}\begin{enumerate}
\item[i)]
$\displaystyle W^{(\bar\jmath)}_{\bar\imath}\bot 
W^{(\bar\jmath)}_{\bar\imath'}$
for all $\bar\imath,\bar\imath',\bar\jmath\in\{0,1\}^k$, $\bar\imath\neq\bar\imath'$.
\item[ii)]
$\displaystyle W^{(j_1,\ldots,j_{k+1})}_{i_1,\ldots,i_{k+1}}
\subseteq W^{(j_1,\ldots,j_{k})}_{i_1,\ldots,i_{k}}$
\item[iii)]
For any choice of 
$\displaystyle U_{\bar\imath}\in\big\{W^{(\bar\jmath)}_{\bar\imath}:\bar\jmath\big\}$,
it holds $\bigcap_{\bar\imath} U_{\bar\imath}=\{0\}$.
\end{enumerate}}
\item[c)]
Let $A:\IC^n\to\IC^n$ denote an hermitian linear map
with $2m$--fold degenerate eigenvalue $\lambda\in\IR$.
Let $W^{(j)}_i$ denote $m$-dimensional
subspaces of $\kernel(A-\lambda\id)$ according
to Item~a), i.e. 
such that $W^{(j)}_0\bot W^{(j)}_1$ and
$W^{(0)}_0\cap W^{(1)}_1=\{0\}=W^{(0)}_1\cap W^{(1)}_0$.
Then to every sufficiently small $\epsilon>0$ and $j=0,1$ there is a
hermitian linear map $A^{(j)}:\IC^n\to\IC^n$ with
{\setlength{\parskip}{-9pt}\begin{itemize}
\item $\displaystyle A^{(j)}\big|_{\kernel(A-\lambda\id)^\bot}\;\equiv\;
A\big|_{\kernel(A-\lambda\id)^\bot}$
\item $\|A-A^{(j)}\|\leq\epsilon$;
\item $A^{(j)}$ is real if $A$ was.
\item $\displaystyle A^{(j)}\big|_{\kernel(A-\lambda\id)}$
 has eigenspaces $W^{(j)}_0$ and $W^{(j)}_1$
\item to different eigenvalues distinct from those of $A$
\end{itemize}}
In particular, eigenvectors of $A^{(j)}$ lie in
$\kernel(A-\lambda\id)^\bot\cup W^{(j)}_0\cup W^{(j)}_1$.
\end{enumerate}
\end{myremark}

\begin{proof}[Theorem~\ref{t:EV}]
Theorem~\ref{t:EV1} shows that $\EVec_n$ is
$(d+1)$-computable for $d:=\lfloor\log_2 n\rfloor$.
It remains to show the existence of a witness of $d$-discontinuity, and w.l.o.g. $n=2^d$.
To this end, start with $A:=0$ and $A^{(j)}_{N}$ ($j=0,1$)
according to Remark~\ref{r:EV}c) with $\epsilon:=1/N$, 
$\lambda:=0$, $m:=2^{d-1}$.
It follows $\EVec(A^{(j)})\subseteq W^{(j)}_0\cup W^{(j)}_1$,
thus $\EVec(A^{(0)})\cap\EVec(A^{(1)})=\emptyset$;
hence 
$\ball\big(\tilde\EVec(A^{(0)}),\delta\big)\cap\ball\big(\tilde\EVec(A^{(1)}),\delta\big)
=\emptyset$ for some $\delta>0$ according to
Lemma~\ref{l:LinEq}b) where 
$\tilde\EVec:=\norm\circ\EVec:\IR^{\binom{n}{2}}\to\calS^{n-1}$
has compact range and, by virtue of Lemma~\ref{l:Comp2}b),
the same complexity of non-uniform computability as $\EVec$.
This yields, according to Lemma~\ref{l:Simpler}b),
(in a complicated way) a witness of discontinuity of $\EVec$.
\\
Now iterating Remark~\ref{r:EV}c) with the subspaces
according to Remark~\ref{r:EV}b), we obtain
matrix sequences $A^{(j_1,\ldots,j_k)}_{n_1,\ldots,n_k}$ 
for $1\leq k\leq d$ and $j_1,\ldots,j_k\in\{0,1\}$
and $n_1,\ldots,n_k\in\IN$ with 
$\displaystyle
\EVec(A^{(\bar\jmath)}_{\bar n})\subseteq
\bigcup\nolimits_{\bar\imath\in\{0,1\}^k} W^{(\bar\jmath)}_{\bar\imath}$
for each $\bar\jmath\in\{0,1\}^k$; hence 
\[ \bigcap\nolimits_{\bar\jmath\in\{0,1\}^k}
\EVec(A^{(\bar\jmath)}_{\bar n})\quad\subseteq\quad
\bigcap\nolimits_{\bar\jmath}\bigcup\nolimits_{\bar\imath}
W^{(\bar\jmath)}_{\bar\imath}
\quad=\quad\{0\} \]
by Remark~\ref{r:EV}b(i-iii).
Therefore $\bigcap\nolimits_{\bar\jmath\in\{0,1\}^k}
\ball\big(\tilde\EVec(A^{(\bar\jmath)},\epsilon\big)=\emptyset$
for some $\epsilon>0$ according to Lemma~\ref{l:LinEq}b).
So Lemma~\ref{l:Simpler}d) finally yields
a witness of $d$-discontinuity.
\qed\end{proof}

\subsection[Root Finding: 
Countably Infinite Complexity]{Root Finding}
We now address the effective Intermediate Value Theorem
\mycite{Theorem~6.3.8.1}{Weihrauch}.
Closely related is the problem of selecting
from a given closed non-empty interval
some point, recall Example~\ref{x:Roux}d).
Both are treated quantitatively
within our complexity-theoretic framework.

Specifically concerning Example~\ref{x:Roux}d),
observe that any non-degenerate interval $[a,b]$ contains
a rational (and thus computable) point $x$;
and providing an integer numerator and denominator of $x$
makes the problem of computably \emph{selecting} some $x$
from given $[a,b]$ trivial. On the other hand, rational
numbers may require arbitrarily large descriptions;
even more, there are intervals containing rationals
only of such large Kolmogorov Complexity;
cf. Claim~d) of the following

\begin{myremark} \label{r:Liouville}
\begin{enumerate}
\item[a)]
There exists an unbounded function $\varphi:\IN\to\IN$
such that the Kolmogorov Complexity $C(m)$ of any 
integer $m\geq n$ is at least $\varphi(n)$.
\item[b)]
Fix $q\in\{2,3,\ldots\}$. 
For $x\in [0,1)\cap\IQ$, $x=r/s=\sum_{i=1}^\infty a_i q^{-i}$, 
$r,s\in\IZ$ coprime and $a_i\in\{0,1,\ldots,q-1\}$,
$C(r,s)$ and $C\big((a_i)_{_i}\big)$ agree up to some
constant independent of $x$ (but possibly depending on $q$).
\end{enumerate}
That is, just like the usual Kolmogorov Complexity
(of a binary string or integer)
depending up to an additive constant 
on the universal machine under consideration
\mycite{Theorem~2.1.1}{Vitanyi}, the complexity $C(x)$
of a \emph{rational} number $x\in\IQ$ 
is well-defined up to $\pm\calO(1)$.
\begin{enumerate}
\item[c)]
For $a,b\in\IQ$, 
$C(a+b),C(a-b),C(a\cdot b),C(a/b)\leq C(a)+C(b)+\calO(1)$. \\
To every $a\in\IQ$ there exists $\bar\sigma\in\Sigma^\omega$
with $C(a)\leq C(\bar\sigma)+\calO(1)$,
the latter in the sense of Proposition~\ref{p:Kolmogorov}.
\item[d)]
Let $x\in\IR$ be algebraic of degree 2
(e.g. $x=\sqrt{p}+q$ for some prime number $p\in\IP$ and $q\in\IQ$).
Then there exists $\varepsilon>0$ such that
for all $r,s\in\IZ$ with $s>0$,
$|x-r/s|>\varepsilon/s^2$.
\item[e)]
Given $N\in\IN$, there exist $a,b\in\IQ\cap[0,1]$
such that all $x\in\IQ\cap[a,b]$ have
$C(x)\geq N$.
\end{enumerate}
\end{myremark}
\begin{proof}
\begin{enumerate}
\item[a)] is from \cite[\textsc{Theorem~2.3.1}i]{Vitanyi}.
\item[b)] On the one hand, a constant-size program can
easily convert $(r,s)$ to the sequence $(a_i)$;
hence $C\big((a_i)_{_i}\big)\leq C(r,s)+c$.
Concerning a converse inequality, 
$x=\sum_{i=1}^\infty a_i q^{-i}\in\IQ$
implies that $(a_i)_{_i}$ be periodic after some initial segment;
i.e. $a_{i}=a_{i+n}=a_{i+2n}=\ldots$ for all $i\geq m$;
hence 
\begin{eqnarray*} x &=& \underbrace{\sum\nolimits_{i<m} a_i q^{-i}}_{=:u} + 
\Big(\underbrace{\sum\nolimits_{i=0}^{n-1} a_{m+i} q^{n-i}}_{=:v}\Big)\cdot\big(
q^{-m-n}+q^{-m-2n}+q^{-m-3n}+\cdots\big) \quad= \\
 &=& u+v\cdot q^{-m}/(q^n-1)
\end{eqnarray*}
with $u\in\IQ$ and $v,q^{-m}/(q^n-1)\in\IZ$ which
can easily be converted into coprime $r,s$ with $x=r/s$.
Also both $m$ (the length of the initial segment) and $n$
(the period) need not be stored separately but can be sought for 
computationally within the sequence $(a_i)$.
\item[c)]
It is easy, and uses only constant size overhead, 
to combines Turing machines computing $a$ and $b$
into ones computing $a+b$, $a-b$, $a\times b$, and $a/b$, respectively.
Moreover a machine computing numerator and denominator of $a$
can be adapted to calculate a $\myrho$--name $\bar\sigma$ of $a$.
\item[d)] is \textsf{Liouville's Theorem} on Diophantine approximation.
\item[e)] 
Take $x\in(1/3,2/3)$ algebraic of degree 2,
$\varepsilon>0$ according to c). Choose $0<\delta<1/3$
such that $\varphi(\sqrt{\varepsilon/\delta})>N$.
Then $\IQ\ni r/s\in[x-\delta,x+\delta]$ requires
$\varepsilon/s^2<\delta$, hence $s>\sqrt{\varepsilon/\delta}$
and $C(x)=C(r,s)\geq C(s)\geq N$.
\qed\end{enumerate}\end{proof}
Note that Remark~\ref{r:Liouville}e) applies only to
rational numbers; that is $[a,b]$ might still contain, say,
algebraic reals $x$ of low Kolmogorov complexity.
We now extend the claim to computable elements:
Referring to Proposition~\ref{p:Kolmogorov},
Theorem~\ref{t:Roux}b) below shows that, 
even with the help of negative information about
(i.e. a $\psiG{}$--name of) a given interval $[a,b]$,
unbounded discrete advice is in general necessary
to find (a $\myrho$--name of) some $x\in[a,b]$.

\begin{theorem} \label{t:Roux}
\begin{enumerate}
\item[a)]
Finding a zero of a given continuous function $f:[0,1]\to[-1,+1]$
with $f(0)=-1$ and $f(1)=+1$, that is the multivalued mapping
$\Intermed:\subseteq C[0,1]\toto[0,1]$, 
\[ f\mapsto f^{-1}[0] \text{ on }
\dom(\Intermed):=\big\{f:[0,1]\to[-1,+1]\text{ continuous},
f(1)=1=-f(0)\big\} \enspace ,\]
has $\mycomp(\Intermed,[\myrho\myto\myrho],\myrho)=\mycard(\Intermed)=\omega$.
\item[b)]
Selecting some point from a given co-r.e. closed bounded non-degenerate interval,
specifically the multivalued mapping 
\[ \Select:[a,b]\mapsto [a,b], \qquad
\dom(\Select):=\{[a,b] : 0\leq a<b\leq 1\} \enspace , \]
is not $d$-wise $(\psiG{},\myrho)$--continuous 
for any $d\in\IN$.
\end{enumerate}
\end{theorem}
Discontinuity of $\Intermed$ is well-known due to, and to occur for, 
arguments $f$ which `hover' \mycite{Theorem~6.3.2}{Weihrauch}. 
We iterate this property
to obtain a witness of $d$-discontinuity for
arbitrary $d\in\IN$:

\begin{figure}[htb]
\centerline{\includegraphics[width=0.99\textwidth]{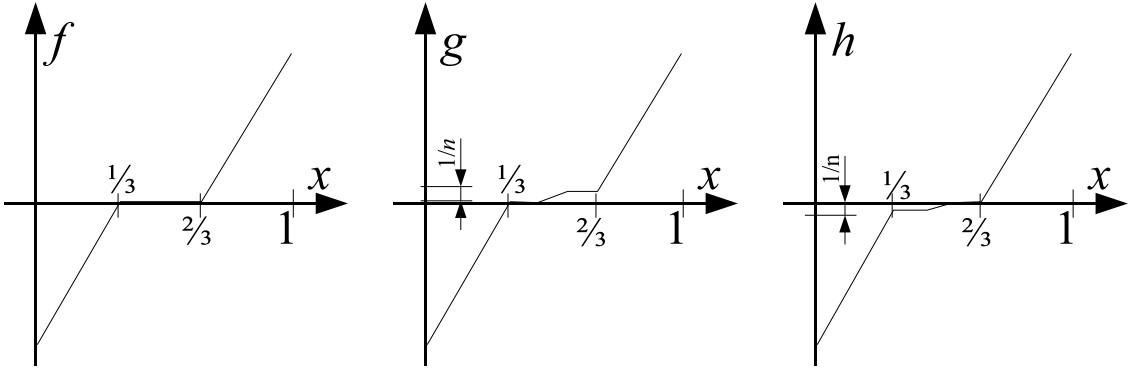}}%
\caption{\label{f:intermed}Witness of discontinuity
for the effective Intermediate Value Theorem.}
\end{figure}

\begin{myremark} \label{r:Intermed}
\begin{enumerate}
\item[a)]
Consider the piecewise linear, continuous function $f:[0,1]\to[-1,+1]$,
\[ f(x):=3x-1
\text{ for } x\in[0,1/3], \quad f:\equiv 0 \text{ on } [1/3,2/3],
\quad f(x):=3x-2 \text{ for } x\in[2/3,1] \enspace . \]
Then $A:=[1/3,5/12]$ and $B:=[7/12,2/3]$ lie in $f^{-1}[0]=[1/3,2/3]=:I$.
Moreover to $n>0$ there are piecewise linear continuous functions
$g,h\in\dom(\Intermed)$ with $g^{-1}[0]=A$ and $h^{-1}[0]=B$
and $\|f-g\|_\infty,\|f-h\|_\infty<1/n$;
cf. Figure~\ref{f:intermed}.\\
In particular, $\ball\big(\Intermed(g),1/27\big)\cap\ball\big(\Intermed(h),1/27\big)=\emptyset$.
\item[b)]
Let $f^{()}:=f$, $I^{()}:=I$,
$f^{(0)}_n:=g$, $I^{(0)}:=A$,
$f^{(1)}_n:=h$, $I^{(1)}:=B$. \\
Iterating the above construction we obtain,
to every $d\in\IN$ and $(i_1,\ldots,i_d)\in\{0,1\}^d$,
closed intervals $I^{(i_1,\ldots,i_d)}$
{\setlength{\parskip}{-9pt}\begin{enumerate}
\item[i)] of length $3^{-d-1}$
\item[ii)] with $I^{(i_1,\ldots,i_{d-1},i_d)}\subseteq I^{(i_1,\ldots,i_{d-1})}$
\item[iii)] such that $\ball\big(I^{(i_1,\ldots,i_{d-1},0)},3^{-d-2}\big)\cap
\ball\big(I^{(i_1,\ldots,i_{d-1},1)},3^{-d-2}\big)=\emptyset$.
\end{enumerate}}
and sequences of functions $f^{(i_1,\ldots,i_d)}_{n_1,\ldots,n_d}\in\dom(\Intermed)$
with 
{\setlength{\parskip}{-9pt}\begin{enumerate}
\item[iv)] $f^{(i_1,\ldots,i_k)}_{n_1,\ldots,n_k}=
\ulim_m f^{(i_1,\ldots,i_k,\ldots,i_d)}_{n_1,\ldots,n_k,m,\ldots,m}$
\item[v)] and $\big(f^{(i_1,\ldots,i_d)}_{n_1,\ldots,n_d}\big)^{-1}[0]=I^{(i_1,\ldots,i_d)}$
\end{enumerate}}
where $\phi=\ulim_m \phi_m$ means 
$\lim_{m\to\infty} \|\phi-\phi_n\|_\infty\to 0$.
\item[c)]
Let $A\subseteq B\subseteq\IR^d$
and $f:\IR^d\to\IR$ continuous
with $f(\vec x)\leq d_B(\vec x)$
for all $\vec x\in\IR^d$.
Then there exists a sequence of continuous
functions $g_\ell:\IR^d\to\IR$ with
$g_1=f$ and $d_A(\vec x)=\sup_\ell g_\ell(\vec x)$: \\
namely $g_\ell:=f\cdot 1/\ell + (1-1/\ell)\cdot d_A$.
\end{enumerate}
\end{myremark}

\begin{proof}[Theorem~\ref{t:Roux}]
\begin{enumerate}
\item[a)]
As has been frequently exploited before
\mycite{Section~6.3}{Weihrauch},
$f\in\dom(\Intermed)$ has an entire
interval of zeros or has some isolated root.
In the latter case, such a root can be found
according to \mycite{Theorem~6.3.7}{Weihrauch}.
In the former case, that interval 
contains a rational one---which
can be provided explicitly by its 
numerator and denominator 
as (unbounded) discrete advice.
We thus have shown $\mycomp(\Intermed)\leq\omega$.\\
Conversely, Remark~\ref{r:Intermed}b) implies in
connection with Lemma~\ref{l:Simpler}d) that
$\Intermed$ is not $d$-continuous
for any $d\in\IN$.
\item[b)]
Consider the 
intervals $I^{(i_1,\ldots,i_d)}$ from Remark~\ref{r:Intermed}b).
However Lemma~\ref{l:Simpler}d) does not apply directly
since the space of closed (non-degenerate) sub-intervals of $[0,1]$
is, equipped with representation $\psiG{}$ rather than with $\psi$
\mycite{Theorem~5.2.9}{Weihrauch}, not metric.
Instead, we resort to 
Lemma~\ref{l:Flag2} as follows: 
Suppose $F$ is a $(\psiG{},\myrho)$--realizer of $\Select$
and recall that a $\psiG{}$--name of closed $\emptyset\not=B\subseteq[0,1]$
is a sequence of continuous functions $g_\ell:[0,1]\to\IR$
with $d_B(x)=\sup_\ell g_\ell(x)$.
By Remark~\ref{r:Intermed}c), any (initial segment of)
such a name can be slightly perturbed to one of 
$A\subseteq B$.
\\
So start with such a sequence $g^{()}=(g_\ell)_{_\ell}$
for $I^{()}$; then take
a sequence of sequences $g^{(0)}_{n,\ell}$ 
with $g^{(0)}_{n,\ell}:=g_\ell$ for $\ell\leq n$
and $\sup_\ell g^{(0)}_{n,\ell}(x)=d_{I^{(0)}}(x)$
according to Remark~\ref{r:Intermed}c);
that is, $(g^{(0)}_{n,\ell})_{_{\ell}}$ is a 
$\psiG{}$--name of $I^{(0)}\subseteq I^{()}$
`resembling' a $\psiG{}$--name $I^{()}$ for $\ell\leq n$.
In particular, $g^{()}_\ell=\ulim_n g^{(0)}_{n,\ell}$
uniformly in $x$ and $\ell$.
Similarly take $g^{(1)}_{n,\ell}$ 
corresponding to $I^{(1)}$
initially resembling $I^{()}$.
Because of Remark~\ref{r:Intermed}b\,iii),
Lemma~\ref{l:Simpler}b) 
yields a witness of discontinuity for $F$.
\\
Now iterate this construction with the
intervals $I^{(n_1,\ldots,n_d)}$ from
Remark~\ref{r:Intermed}b) to obtain
a witness of $d$-discontinuity
according to Lemma~\ref{l:Simpler}c).
\qed\end{enumerate}\end{proof}

\section{Conclusion, Extensions, and Perspectives} \label{s:Conclusion}
We claim that a major source of criticism against
Recursive Analysis misses the point:
although computable functions $f$ \emph{are} necessarily continuous
when given approximations to the argument $x$ \emph{only},
most practical $f$'s do become computable
when providing in addition some discrete information about $x$.
Such `advice' usually consists of some very natural 
and mathematically explicit integer value
from a bounded range (e.g. the rank of the matrix under consideration)
and is readily available in practical applications.

We have then turned this observation into a complexity theory,
investigating the minimum \emph{size} (=cardinal) of the range 
this discrete information comes from.
And we have determined this quantity
for several simple and natural problems from linear algebra: 
calculating the rank of a given matrix,
solving a system of linear equalities,
diagonalizing a symmetric matrix, 
and finding some eigenvector to a given symmetric matrix. 
The latter three are inherently multivalued.
And they exhibit a considerable difference in complexity:
for input matrices of format $n\times n$,
usually discrete advice of order $\Theta(n)$ 
is necessary and sufficient; 
whereas some single eigenvector
can be found using only $\Theta(\log n)$--fold advice:
specifically, the quantity 
$\big\lfloor\log_2\min\big\{\dim\kernel(A-\lambda\id):\lambda\in\sigma(A)\big\}\big\rfloor$.
The algorithm exploits this data based on some 
combinatorial considerations---which nicely complement
the heavily analytical and topological arguments
usually dominant in proofs in Recursive Analysis.

\medskip
Our lower bound proofs assert $d$-discontinuity
of the function under consideration. They can
be extended (yet become even more tedious when trying
to do so formally) to \emph{weak} $d$-discontinuity.
Also the major tool for such proofs, namely that of
witnesses of $d$-discontinuity, would deserve
generalizing from effective metric to computable
topological spaces.

\subsection{Non-Integral Advice} \label{s:Fractional}
Theorem~\ref{t:LinEq} shows that $d$-fold advice does not
suffice for effectively finding a nontrivial solution $\vec x$ to
a homogeneous equation $A\cdot\vec x=0$; whereas $(d+1)$-fold
advice, namely providing $\rank(A)\in\{0,\ldots,d\}$, does suffice.
\begin{itemize}
\item
Since the rank can be effectively approximated from below 
(i.e. is $\myrhol$-computable) \mycite{Theorem~7}{LA},
it in fact suffices to provide complementing upper approximations 
(i.e. a $\myrhog$--name) to $\rank(A)$.
One may say that this constitutes strictly 
less than $(d+1)$-fold information.
\item
Similarly concerning diagonalization of a real symmetric
$n\times n$-matrix $A$,
since the number $\Card\sigma(A)$ 
of distinct eigenvalues can be effectively
approximated from below, it suffices to provide only
complementing upper approximations---cmp. 
\mycite{Theorem~19}{LA}---which may be regarded
as strictly less than $n$-fold advice.
\item
Similarly, with respect to the problem of finding \emph{some}
eigenvector of $A$, again strictly less than
$(\lfloor\log_2 n\rfloor+1)$-fold advice suffices:
namely lower approximations to $\lfloor\log_2 m(A)\rfloor$
(with $m(A)$ from Theorem~\ref{t:EV1})
based on the following
\end{itemize}

\begin{observation} \label{o:MinDim}
The mapping
$\IR^{\binom{n}{2}}\ni A\mapsto\lfloor\log_2 m(A)\rfloor$
is $(\myrho^{n\cdot(n-1)/2},\myrhog)$--computable.
\end{observation}
\begin{proof}
Given $\lambda$,
$\dim\kernel(A-\lambda\id)=n-\rank(A-\lambda\id)$
is $\myrhog$--computable by \cite[\textsc{Theorem~7}i]{LA};
hence so is its minimum $m(A)$ over all $\lambda\in\sigma(A)$,
cmp. \mycite{Proposition~17}{LA} and \mycite{Exercise~4.2.11}{Weihrauch}.
Now although $\log_2:(0,\infty)\ni x\mapsto\ln(x)/\ln(2)\in\IR$
is $(\myrhog,\myrhog)$--computable by monotonicity,
$x\mapsto\lfloor x\rfloor$ is of course not.
On the other hand, $m(A)$ attains only integer values;
and the nondecreasing, purely integral function 
$\IN\ni m\mapsto\lfloor\log_2 m\rfloor\in\IN_0$
is $(\myrhog,\myrhog)$--computable.
\qed\end{proof}
The above examples suggest refining $k$-fold 
advice to non-integral values of $k$:

\begin{definition} \label{d:Fractional}
Let $f:X\to Y$ be a function and $Z$ a topological $\Tnull$ space.
\begin{enumerate}
\item[a)]
Call $f$ \emph{continuous with $Z$-advice} if 
there exists a function $g:X\to Z$ such that
the function $f|^g$, defined as follows, is continuous:
\begin{equation} \label{e:Sierpinski}
\dom(f|^g) \;:=\; \{ (x,z) : x\in X, g(x)=z \}
\;\subseteq\; X\times Z, \qquad
(x,z)\mapsto f(x) \enspace . 
\end{equation}
\item[b)]
Let $Z$ be finite and 
fix some injective notation $\nu_Z:\subseteq\Sigma^*\to\tau$ of 
the (finitely many) open subsets of $Z$. Then the
representation $\delta_{Z}=\delta_{Z,\nu_Z}:\subseteq\Sigma^\omega\to Z$
is defined as follows: \\
$\bar\sigma\in\Sigma^\omega$ is a $\delta_{Z,\nu}$-name of $z$
~iff~ it is a $\nu$-enumeration (with arbitrary repetition)
of all open sets containing $Z$.
\item[c)]
For effective metric spaces $(X,\alpha)$ and $(Y,\beta)$,
call $f$ \emph{$(\alpha,\beta)$-computable with $Z$-advice} if
there exists some $g:X\to Z$ such that the function
$f|^g:\subseteq X\times Z\to Y$ from a) is
$(\alpha\times\delta_Z,\beta)$-computable.
\end{enumerate}
\end{definition}
Restricting $Z$ to discrete spaces, one
recovers Definition~\ref{d:Nonunif}a):

\begin{lemma} \label{l:Fractional}
For $d\in\IN$ let 
$Z_d$ denote the set $\{0,1,\ldots,d-1\}$
equipped with the discrete topology.
\begin{enumerate}
\item[a)] A function $f:X\to Y$ is $d$-continuous 
~iff~ it is continuous with $Z_d$-advice.
\item[b)] The representation $\delta_{Z_d}$ of $Z_d$ 
  is computably equivalent to $\nu_{Z_d}$: $\delta_{Z_d}\equiv\nu_{Z_d}$.
  Whereas in general, $\nu_Z$ is only computably
  reducible to (but not from) $\delta_{Z}$: $\nu_Z\reducneq\delta_Z$.
\item[c)] A function $f:X\to Y$ is $(\alpha,\beta)$-computable
  with $d$-wise advice ~iff~ it is $(\alpha,\beta)$-computable
  with $Z_d$ advice.
\item[d)] 
$(Z,\delta_Z)$ is admissible. In particular
if $(X,\alpha)$ and $(Y,\beta)$ are admissible
and it function  $f:X\to Y$ is
$(\alpha,\beta)$-computable with $Z$-advice,
then $f$ is continuous with $Z$-advice.
\end{enumerate}
\end{lemma}
\begin{proof}
\begin{enumerate}
\item[a)]
Let $\Delta=\{D_0,D_1,\ldots,D_{d-1}\}$ denote an $d$-element partition of $X$
such that $f|_{D_z}$ is continuous for each $z=0,1,\ldots,d-1$.
Define $g:X\ni x\mapsto$the unique $z\in Z_d$ with $x\in D_z$.
Then, for open $V\subseteq Y$, 
\begin{equation} \label{e:Fractional}
(f|^g)^{-1}[V] \quad=\quad 
\biguplus_{z\in Z} \big(f^{-1}[V]\cap \underbrace{g^{-1}[\{z\}]}_{=D_z}\big)\times\{z\}
\quad=\quad \biguplus_{z\in Z} \underbrace{(f|_{D_z})^{-1}[V]}_{\text{open in }D_z}
\times\{z\}
\end{equation}
is relatively open 
in $\biguplus_{z\in Z} D_z\times\{z\}=\dom(f|^g)$, i.e. $f|^g$ continuous.
\\
Conversely let $f|^g$ be continuous for $g:X\to Z_d$.
Define $\Delta:=\{D_0,D_1,\ldots,D_{d-1}\}$ where $D_z:=g^{-1}[\{z\}]$ for $z\in Z_d$.
Then Equation~(\ref{e:Fractional}) requires that
$\bigcup_z \big(f^{-1}[V]\cap D_z\big)\times\{z\}$ 
be open in $\bigcup_z D_z\times\{z\}$.
Now $D_z\times\{z\}$ is open by definition of the product
topology and because $z\in Z_d$ is discrete,
this implies that also the intersection 
$(f|^g)^{-1}[V]\cap (D_z\times\{z\})=(f^{-1}[V]\cap D_z)\times\{z\}$
be open in $D_z\times\{z\}$, i.e. that $f|_{D_z}^{-1}[V]$ is open in $D_z$;
hence $f|_{D_z}$ is continuous for each $z\in Z_k$.
\item[b)]
Since $Z$ is finite and $\nu_Z:\subseteq\Sigma^*\to Z$ is injective, 
everything is bounded a priori. For instance, given a $\nu_Z$-name
of $z\in Z$, one can easily produce a pre-stored list of all 
(finitely many) open sets containing this $z$:
thus showing $\nu_Z\reduceq\delta_Z$.
\\
For the converse, exploit that $Z_d$ bears the discrete topology
and therefore is effectively \Tone \cite{EffT}: 
Given an enumeration of all (finite) open sets $U_i$ 
containing $z\in Z_d$, their intersection $\bigcap_i U_i$
becomes a singleton after finite time, thereby identifying $z$.
\\
In the \Sierpinski space $\Sierp$ from Example~\ref{x:Sierpinski}a) below,
(some $\delta_{\Sierp}$-names of)
$1=\bot$ cannot continuously be distinguished from 
(a $\delta_{\Sierp}$-name of) $0=\top$.
\item[c)]
Suppose that $f|_{D_z}$ be $(\alpha,\beta)$-computable 
for each $z\in Z_d$. Since $Z_d$ is finite, it follows
that $f|^g:(x,z)\mapsto f|_{D_z}(x)$ is 
$(\alpha\times\nu_{Z_d},\beta)$-computable
and hence $(\alpha\times\delta_{Z_d},\beta)$-computable by b).
\\
Conversely let $f|^g$ be $(\alpha\times\delta_{Z_d},\beta)$-computable.
Then, similarly to a) and since each $z\in Z_d$ 
is $\delta_{Z_d}$-computable, 
it follows that also the all restrictions
$f|_{D_z}$ be $(\alpha,\beta)$-computable for $z\in Z_d$
where $D_z=g^{-1}[\{z\}]$. 
\item[d)]
Observe that $\delta_Z$ coincides with the \emph{standard representation}
of the (finite, hence effective) $\Tnull$-space $Z$; compare
\mycite{Section~3.2}{Weihrauch}. 
Concerning the second claim, 
\mycite{Corollary~3.2.12}{Weihrauch} reveals that
$f|^g$ is continuous.
\qed\end{enumerate}\end{proof}
In this sense, Example~\ref{x:Staircase} turns out
to suffice with even strictly less\footnote{We
refrain from defining reducibility between general $\Tnull$-spaces
but refer to Example~\ref{x:Sierpinski}d) for the specific spaces
$Z_d$ and $\Sierp_d$, $d\in\IN$.} than 2-fold advice:

\begin{myexample} \label{x:Sierpinski}
\begin{enumerate}
\item[a)]
Consider as $Z$ the \Sierpinski space
$\Sierp$, i.e. the set $\{0,1\}$ equipped with the 
topology $\big\{\emptyset,\{0\},\{0,1\}\big\}$ as
open sets.
Then the characteristic function of the 
\emph{complement} of the Halting problem
$\cf{\IN\setminus H}:\IN\to\Sierp$ is $(\nu,\delta_{\Sierp}$)-computable,
but $\cf{H}$ itself is not.
\item[b)]
The Gau\ss{} Staircase function 
$f:=\lfloor\,\cdot\,\rfloor:\IR\to\IZ$
is $(\myrho,\myrho)$-computable with $\Sierp$-advice.
\item[c)]
Generalizing $\Sierp=:\Sierp_1$, 
denote by $\Sierp_d$ the set $\{0,1,\ldots,d\}$ equipped with the 
following topology:
$\big\{\emptyset,\{0\},\{0,1\},\{0,1,2\},\ldots,\{0,1,2,\ldots,d\}\big\}$.
Then $\rank:\IR^{d\times d}\to\{0,1,\ldots,d\}$ is 
$(\myrho^{d\times d},\myrho)$-computable with $\Sierp_d$-advice.
\item[d)]
For each function $f$ and integer $d$,
continuity/computability with $\Sierp_{d-1}$-advice
implies continuity/computability with $\Sierp_{d}$-advice,
implies continuity/computability with $Z_{d+1}$-advice,
implies continuity/computability with $Z_{d+2}$-advice.
\\
On the other hand, the Dirichlet Function
$\cf{\IQ}:[0,1]\subseteq\IR\to\{0,1\}$
is computable with $Z_2$-advice
but not continuous with $\Sierp_d$-advice
for any $d\in\IN$.
\end{enumerate}
\end{myexample}
In view of \mycite{Theorem~11}{LA},
Example~\ref{x:Sierpinski}c) shows that $\Sierp_{d}$-advice renders
also $\LinEq_{n,m}$ computable for $d:=\min(n,m-1)$.
\begin{proof}[Example~\ref{x:Sierpinski}]
\begin{enumerate}
\item[a)]
Simulate the given Turing machine $M$, and for each step
append ``$\{0,1\}$'' to the $\delta_{\Sierp}$-name of 
$1=\cf{\IN\setminus H}(\langle M\rangle)$ to output in case
that $M$ does not terminate; whereas if and when $M$ does
turn out to terminate, start appending ``$\{0\}$'' to the output,
thus indeed producing a $\delta_{\Sierp}$-name of $0$.
\\
Since any $\delta_{\Sierp}$-name of $0$ must include the set
$\{0\}$ in its enumeration, one can distinguish it in finite
time from a $\delta_{\Sierp}$-name of $1$. 
$\delta_{\Sierp}$-computing $\cf{H}(\langle M\rangle)=0$ would
thus amount to detecting the non-termination of $M$,
contradicting that $H$ is not co-r.e.
\item[b)]
Intuitively, ``$x\not\in\IZ$'' is semi-decidable;
hence suffices to provide only half-sided advice 
for the case ``$x\in\IZ$''. Formally,
define $g:\IR\to\Sierp$ as the characteristic function of $\IR\setminus\IZ$,
i.e., $g:\IZ\ni x\mapsto 0$ and 
$g:\IR\setminus\IZ\ni x\mapsto1$.
Observe $\dom(f|^g)=\big(\IZ\times\{0\}\big)\uplus
\big(\IR\setminus\IZ\times\{1\}\big)$.
Hence, for $y\in\IN$, it is
$\big(f|^g\big)^{-1}[\{y\}]
=\big((y,y+1)\times\{1\}\big)\uplus\big(\{y\}\times\{0\}\big)$
with $(y,y+1)\times\{1\}=\big((y,y+1)\times\Sierp\big)
\cap\dom(f|^g)$ 
and $\{y\}\times\{0\}=\big(IR\times\{0\}\big)\cap\dom(f|^g)$
both open in $\dom(f|^g)$.
\\
$\Sierp\reduceq Z_2$ can be seen from the mapping
$i:Z_2\to\Sierp$ with $0\mapsto 0$ and $1\mapsto 1$
being trivially continuous since $Z_2$ has the discrete topology.
Conversely, both surjective mappings from $\Sierp$ to $Z_2$
are discontinuous; hence $Z_2\not\reduceq\Sierp$.
\item[c)]
The identity mapping $\id:Z_{k+1}\{0,1,\ldots,k\}\to\{0,1,\ldots,k\}=\Sierp_k$,
is surjective and trivially continuous, hence $\Sierp_k\reduceq Z_{k+1}$ holds.
Similarly, $\Sierp_{k-1}\reduceq\Sierp_{k}$ is established 
by the surjection $h_k:\Sierp_k\to\Sierp_{k-1}$ 
defined as $0<i\mapsto i-1$ and $0\mapsto 0$,
whose continuity follows from
$h_k^{-1}[\{0,1,\ldots,i\}]=\{0,1,\ldots,i,i+1\}$ for $0\leq i<k$.
\\
Let $g:=\rank:\IR^{k\times k}\to\Sierp_k$.
Then it holds $\rank|^{\rank}(A,i)=\rank(A)=i$ on
$\dom(\rank|^{\rank})=\{(A,i):A\in\IR^{k\times k},i\in\IN,\rank(A)=i\}$.
In particular $\big(\rank|^{\rank}\big)^{-1}[\{j\}]
=\big(\big\{A:\rank(A)\geq j\}\times\{0,1,\ldots,j\}\big)\cap\dom(\rank|^{\rank})$
is relatively open in $\dom(\rank|^{\rank})$,
because $\{A:\rank(A)\geq j\}\subseteq\IR^{k\times k}$ 
is open by \mycite{Theorem~7}{LA} and $\{0,1,\ldots,j\}\subseteq\Sierp_k$
is open by definition.
\item[d)]
A function $g:X\to\Sierp_{d-1}$ is also one $g:X\to\Sierp_d$;
and each open subset of $\Sierp_{d-1}$ is also open in $\Sierp_d$;
plus $\delta_{\Sierp_{d-1}}\reduceq\delta_{\Sierp_d}$ holds:
Therefore continuity/computability of 
$f|^g:\subseteq X\times\Sierp_{d-1}\to Y$
implies continuity/computability of $f|^g:\subseteq X\times\Sierp_{d}\to Y$.
Similarly for $Z_{d+1}$ with $Z_{d+2}$.
Moreover, a function $g:X\to\Sierp_d$ can be considered
as a function $g:X\to Z_{d+1}$; and each open subset of
$\Sierp_d$ is also open in $Z_{d+1}$; 
plus $\delta_{\Sierp_d}\reduceq\delta_{Z_{d+1}}$ holds: 
Therefore continuity/computability of $f|^g:\subseteq X\times\Sierp_{d}\to Y$
implies continuity/computability of $f|^g:\subseteq X\times Z_{d+1}\to Y$.
\\
Now suppose $\cf{\IQ}|^g:\subseteq\IR\times\Sierp_d\to\{0,1\}$ 
is continuous for $g:\IR\to\Sierp_d$. By Item~c)
and Lemma~\ref{l:Fractional}c), $\cf{\IQ}|_{D_z}$ is continuous
(i.e. constant) on each $D_z:=g^{-1}[\{z\}]$; that is, for each 
$z=0,1,\ldots,d$, it either holds
$D_z\subseteq\IQ$ or $D_z\subseteq\IR\setminus\IQ$.
First observe that there exist $k,\ell$ and 
two sequences $x_n\in D_k$ of rationals and 
$y_n\in D_\ell$ of irrationals with
$|x_n-y_n|\to0$. Indeed, $y_n:=y$ arbitrary irrational
belongs to $D_\ell$ for $\ell:=g(y)$; and, since $\IQ$
is dense, there exists $(x_n)\subseteq\IQ$ with $x_n\to y$;
where, by pigeon-hole, $x_n\in D_k$ for some $k$ and 
infinitely many (by proceeding to a subsequence w.l.o.g. all) $n$.
We treat the case $k<\ell$ ($k>\ell$ works similarly).
By construction, it holds $\cf{\IQ}(y_n)=0$ and $g(y_n)=\ell$;
hence $(y_n,\ell)\in(\cf{\IQ}|^g)^{-1}[(-\tfrac{1}{2},+\tfrac{1}{2})]=:V$
for all $n$.
Since $V$ is open in $\dom(\cf{\IQ}|^g)\subseteq\IR\times\Sierp_d$,
it follows $(x_n,k)\in V$
for all sufficiently large $n$; recall that the topology on
$\Sierp_d$ has $k\in U$ for open $U\subseteq\Sierp_d$ and $k<\ell\in U$.
But $\cf{\IQ}(x_n)=1$ contradicts $(x_n,k)\in V$.
\qed\end{enumerate}\end{proof}
Since $\Sierp$-advice is strictly less than $2$-fold advice
(Lemma~\ref{l:Fractional} plus Example~\ref{x:Sierpinski}d),
and $\Sierp$ is strictly richer than 1-fold (i.e. no) advice
(Example~\ref{x:Sierpinski}b), it is consistent to quantify
$\Sierp$-advice as $(1+\tfrac{1}{2})$-fold. In fact, $Z_2$
and $\Sierp$ are (up to homeomorphism) the only 2-element 
$\Tnull$ spaces; but according to the second part of
Example~\ref{x:Sierpinski}d), $\Sierp_d$-advice is \emph{not} 
more (nor less, for $d\geq2$) than $Z_2$-advice
and hence cannot justly be called $(d+\tfrac{1}{2})$-fold.

In fact the Definition~\ref{d:Nonunif} of integral and cardinal
$k$-continuity has an important structural advantage: the complexities 
of two functions are always comparable---either $\mycard(f)<\mycard(g)$,
or $\mycard(f)>\mycard(g)$, or $\mycard(f)=\mycard(g)$;
Whereas when refining beyond integral advice,
non-comparability emerges. In fact Definition~\ref{d:Fractional}
has been suggested to be related to \textsf{Weihrauch Degrees} 
with their complicated structure \cite{Weihrauch92,Pauly,Guido}.

On the other hand, \person{Arno Pauly} has recently suggested 
(private communication during \textsf{CCA2009}) that
at least some of the above lower bound proofs based on the 
technical and (particularly notationally) cumbersome tools
of \emph{witness of discontinuity} can be simplified 
considerably by Weihrauch-reduction 
\mycite{Theorem~5.9}{Pauly} from (some appropriate product of) 
the function $\MLPO_n$ \mycite{Section~5}{Weihrauch92}.

\begin{myquestion}
By assigning \emph{weights} to the advice values $z\in Z$
and to the measurable subsets of $\dom(X)$,
can one obtain a notion of \emph{average advice} in the
spirit of \person{Shannon}'s \textsf{entropy}?
\end{myquestion}

\subsection{Topologically Restricted Advice} \label{s:Topology}
Definition~\ref{d:Nonunif} asks for the number of colour classes
needed to make $f$ continuous/computable on each such class---unconditional
to the topological complexity of the classes themselves:
in principle, they may be arbitrarily high on the Borel Hierarchy
or even non-measurable (subject to the axiom of choice).

From our point of view, 
  determining the discrete advice to
  (i.e. the colour $c$ of) some input $x$ to $f$ is a non-computational
  process preceeding the evaluation of $f$. 
  For instance in the Finite Element Method approach to
  solving a partial differential equation on some surface $S$, 
  its discretization via triangulation gives rise to a matrix $A$
  known a-priori to have 3-band form: its band-width need not be
  `computed', nor does one have to explicitly represent 
  the subset of all 3-band matrices within the collection of
  all matrices.
In fact, since the optimal colour classes themselves 
(rather than the number of colours) is usually far from unique,
this freedom may be exploited to choose them not too wild.

On the other hand, Definition~\ref{d:Nonunif} can easily be 
adapted to take into account topological restrictions:

\begin{definition}
Let $f:A\to B$ denote a function between topological spaces $A,B$
(represented spaces $(A,\alpha)$ and $(B,\beta)$); 
and let $\calA\subseteq2^A$ 
denote a class of subsets of $A=\dom(f)$.
\begin{enumerate}
\item[a)]
$\displaystyle \mycard(f;\calA):=\min\big\{\Card(\Delta):
\Delta\subseteq\calA\text{ partition of }A, f|_D\text{ is continuous }\forall D\in\Delta\big\}$
\item[b)]
$\displaystyle \mycomp(f,\alpha,\beta;\calA):=\min\big\{\Card(\Delta):
\Delta\subseteq\calA\text{ partition of }A, f|_D\text{ is }
(\alpha,\beta)\text{-computable } \forall D\in\Delta\big\}$
\end{enumerate}
\end{definition}
Hence for $\calA:=2^{\dom(f)}$ the powerset of $\dom(f)$
one recovers the previous, unrestricted Definition~\ref{d:Nonunif}:
$\mycard(f;2^{\dom(f)})=\mycard(f)$ and
$\mycomp(f,\alpha,\beta,2^{\dom(f)})=\mycomp(f,\alpha,\beta)$;
whereas restricting the topology of the colour classes may 
increase, but not decrease, the number of colours needed:
$\mycard(f;\calA)\leq\mycard(f,\calB)$ for $\calB\subseteq\calA$;
similarly for $\mycomp$.
Also notice that, unless $\calA$ is closed under finite
unions and intersections, it may now well matter whether
$\Delta$ is a partition or a covering of $\dom(f)$.

Concerning the applications considered in Section~\ref{s:Applications},
Corollary~\ref{c:Comparison2} below shows
that the optimal advice (namely the matrix rank and the 
number of distinct eigenvalues) gives rise to topologically
very tame colour classes. In order to formalize this claim, 
recall that for a metrizable space $X$, each level of 
the Borel Hierarchy $\BorelS_t(X),\BorelP_t(X)\subseteq 
\BorelS_t(X)\cup\BorelP_t(X)
\subseteq\BorelS_{t+1}(X)\cap\BorelP_{t+1}(X)$
of open/closed ($t=1$) set, $\Fsigma/\Gdelta$ ($t=2$)
sets and so on, is strictly refined by the 
\textsf{Hausdorff difference hierarchy};
whose second level $2\Hausd\BorelS_t(X)=2\Hausd\BorelP_t(X)$ consists 
of all sets of the form $U\setminus V$ with $U,V\in\BorelS_t(X)$
(equivalently: of the form $A\setminus B$ with $A,B\in\BorelP_1(X)$)
\mycite{Section~22.E}{Kechris}.
We can now strengthen Proposition~\ref{p:Comparison}a+c):

\begin{lemma} \label{l:Topology}
\begin{enumerate}
\item[a)] Let $X$ be a metrizable space
and $f:X\to Y$. Then, in addition to the inequalities
$\mycard(f;2\Hausd\BorelS_2)\leq\mycard(f;2\Hausd\BorelS_1)$
and $\Lev'(f)\leq\Lev(f)$, it also holds
$\mycard(f;2\Hausd\BorelS_2)\leq\Lev'(f)$
and $\mycard(f;2\Hausd\BorelS_1)\leq\Lev(f)$.
\item[b)]
 The \textsf{Dirichlet Function}, i.e. the characteristic
 function $\cf{\IQ}:[0,1]\subseteq\IR\to\{0,1\}$, has  
 $\mycomp(\cf{\IQ},\myrho,\myrho;2\Hausd\BorelS_2)=\mycard(\cf{\IQ};2\Hausd\BorelS_2)=2$
 but $\Lev'(\cf{\IQ})=\Lev(\cf{\IQ})=\infty$.
\item[c)] Let $f:X\to Y$ be such that $f|_U$ is continuous on open $U\subseteq X$.
  Then it holds $\LEV(f,1)\subseteq X\setminus U$;
  and the prerequisite that $U$ be open is essential.
\item[d)]
  More generally, if $U_{i}\subseteq X\setminus(U_1\cup\cdots\cup U_{i-1})$
  is relatively open and $f|_{U_{i}}$ continuous thereon
  for all $i\leq k$, then $\LEV(f,k)\subseteq X\setminus(U_1\cup\cdots\cup U_k)$.
\end{enumerate}
\end{lemma}
\begin{proof}
\begin{enumerate}
\item[a)]
Recall that $f$ is continuous on $\LEV(f,i)\setminus\LEV(f,i+1)$
where $\LEV(f,0)=\dom(f)$ and 
$\LEV(f,i+1)=\closure{\{x\in\LEV(f,i):f|_{\LEV(f,i)}\text{ discontinuous at }x\}}$
is closed, i.e. belongs to $\BorelP_1\big(\LEV(f,i)\big)\subseteq\BorelP_1\big(\dom(f)\big)$;
hence $\LEV(f,i)\setminus\LEV(f,i+1)\in 2\Hausd\BorelP_1$ constitutes a
$\Lev(f)$--element partition of $\dom(f)$ as required.
\\
For the case of $\Lev'(f)$, recall that the set
$\LEV'(f,i+1)$ of discontinuities of $f|_{\LEV'(f,i)}$ is always
$\Fsigma$, i.e. in $\BorelS_2\big(\LEV'(f,i)\big)$; and by induction
in $\BorelS_2\big(\dom(f)\big)$ since $\Fsigma$ sets are closed under
finite intersection. Now proceed as above.
\item[b)]
Observe that, since $\IQ\in\Fsigma$,
both $\IQ\cap[0,1]$ and $[0,1]\setminus\IQ$
belong to $2\Hausd\BorelS_2$, thus showing
$\mycomp(\cf{\IQ};2\Hausd\BorelS_2)=2$.
\\
However the subset $\LEV'(f,1)$ of discontinuities 
  of $\cf{\IQ}$ coincides with $[0,1]=\LEV'(f,0)$; therefore it holds
  $[0,1]=\LEV'(f,k)=\LEV(f,k)\neq\emptyset$ for all (even transfinite) $k$.
\item[c)] From \mycite{Lemma~2.5.3}{Hertling}, it follows
  that $U$ is disjoint from $\LEV'(f,1)$, i.e.
  $\LEV'(f,1)\subseteq X\setminus U$ a closed set;
  therefore $\LEV(f,1)=\closure{\LEV'(f,1)}$,
  the least closed set containing $\LEV'(f,1)$,
  is a subset of $X\setminus U$.
\\
  Recall from b) the example of $\cf{\IQ}:[0,1]\to\{0,1\}$ 
  continuous on $\IQ$, yet $\IQ$ is certainly not disjoint 
  from $\LEV(\cf{\IQ},1)=\LEV'(\cf{\IQ},1)=[0,1]$.
\item[d)] proceeds by induction on $k$,
  the case $k=1$ been handled in c).
First observe that $\LEV(f,k+1)=\LEV(f|_{\LEV(f,k)},1)$
since the topological closure implicit on the left hand side 
coincides with the closure relative to (closed) $\LEV(f,k)$
on the right hand side. Moreover, the induction hypothesis
$\LEV(f,k)\subseteq X\setminus(U_1\cup\cdots\cup U_k)$ implies 
$\LEV(f|_{\LEV(f,k)},1)\subseteq\LEV(f_{X\setminus(U_1\cup\cdots\cup U_k)},1)$
by \mycite{Lemma~2.5.4}{Hertling},
which is in turn contained in 
$\big(X\setminus(U_1\cup\cdots\cup U_k)\big)\setminus U_{k+1}$
according to c).
\qed\end{enumerate}\end{proof}
Lemma~\ref{l:Topology}a+b) indicates that
the greedy meta-algorithm underlying the definitions of 
$\Lev(f)$ and $\Lev'(f)$ yields topologically mild
colour classes on the one hand, but on the other 
hand not necessarily the least number.
For the problems in linear algebra considered above,
however, greedy is optimal:

\begin{corollary} \label{c:Comparison2}
\begin{enumerate}
\item[a)] Fix $n,m\in\IN$ and recall from Theorem~\ref{t:LinEq}
the problem $\LinEq_{n,m}$ of finding to a given $A\in\IR^{n\times m}$
of $\rank(A)\leq d:=\min(n,m-1)$ some non-zero $\vec x\in\IR^m$ 
such that $A\cdot\vec x\neq0$. It holds
\[ \Lev'(\LinEq_{n,m}) = \Lev(\LinEq_{n,m}) =
\mycard(\LinEq_{n,m}) = 
\mycomp(\LinEq_{n,m},\myrho^{n\times m},\myrho^m;2\Hausd\BorelS_1)
= d+1. \]
\item[b)]
Fix $d\in\IN$ and recall from Theorem~\ref{t:Diag} the problem
$\Diag_{d}$ of finding to a given real symmetric $d\times d$-matrix $A$
a basis of eigenvectors. It holds
\[ \Lev'(\Diag_d)=\Lev(\Diag_d)= \mycard(\Diag_d)
=
\mycomp(\Diag_d,\myrho^{d\cdot(d-1)/2},\myrho^{d\times d};2\Hausd\BorelS_1)
= d
\enspace . \]
\item[c)]
Fix $n\in\IN$ and recall from Theorem~\ref{t:EV} the problem
$\EVec_{n}$ of finding, to a given real symmetric $d\times d$-matrix $A$,
some eigenvector. It holds
\[ \Lev'(\EVec_n)=\Lev(\EVec_n)= \mycard(\EVec_n)
=
\mycomp(\EVec_n,\myrho^{n\cdot(n-1)/2},\myrho^n;2\Hausd\BorelS_1)
=\lfloor\log_2 n\rfloor+1
\enspace . \]
\end{enumerate}\noindent
More precisely, the class $2\Hausd\BorelS_1$ of pairwise differences
of open sets above may be replaced by the class
$2\Hausd\KleeneS_1$ of pairwise differences of
\emph{r.e.} open sets, i.e. by the second Hausdorff level on
the ground level $\KleeneS_1$ of the \emph{effective} Borel Hierarchy.
\end{corollary}
\begin{proof}
\begin{enumerate}
\item[a)]
Theorem~\ref{t:LinEq} refers to arbitrary colour classes
and shows that, there, $d$-fold advice is insufficient to continuity:
$\mycard(\LinEq_{n,m})>d$. In view of Lemma~\ref{l:Topology}a) 
it thus suffices to show
$\Lev(\LinEq_{n,m})\leq d+1$. Indeed,
the set $\rank^{-1}(\geq k)$ of matrices of rank at least $k$
is effectively open a subset of $X:=\IR^{n\times m}$ because
$A\mapsto\rank(A)$ is lower-computable
\cite[\textsc{Theorem~7}i]{LA}. In particular, 
$U_{d-k+1}:=V_{k}:=\rank^{-1}(k)=\rank^{-1}(\geq k)\cap\rank^{-1}(\leq k)$
is effectively open in 
$\rank^{-1}(\leq k)=\dom(\LinEq_{n,m})\setminus(V_{d}\cup\cdots\cup V_{k+1})$;
and $\LinEq_{n,m}$ is computable
and continuous thereon by \mycite{Theorem~11}{LA}.
Now apply Lemma~\ref{l:Topology}d) to conclude 
  $\LEV(\LinEq_{n,m},d+1)\subseteq\dom(\LinEq_{n,m})\setminus(V_d\cup\cdots\cup V_0)
  =\emptyset$.
\item[b)]
Similarly to a) and in view of Theorem~\ref{t:Diag}
it suffices to show $\Lev(\Diag_d)\leq d$. 
Now, again, the set 
$V_k:=\{A:\Card\sigma(A)=k\}=
\{A:\Card\sigma(A)\geq k\}\cap\{A:\Card\sigma(A)\leq k\}$
of symmetric real $d\times d$-matrices $A$ with
exactly $k$ distinct eigenvalues 
is effectively open in 
$\{A:\Card\sigma(A)\leq k\}=\dom(\Diag_d)\setminus (V_d\cup\cdots\cup V_{k+1})$:
because $A\mapsto\Card\sigma(A)$ is lower-computable
and lower-continuous \mycite{Proposition~17}{LA}.
And $\Diag_d$ is computable and continuous on $V_k$ by
\mycite{Theorem~19}{LA}, so Lemma~\ref{l:Topology}d) yields the claim.
\item[c)]
Again,
in order to show $\Lev(\EVec_n)\leq\lfloor\log_2 n\rfloor+1$,
consider the sets
$U_k:=\big\{A\in\IR^{n\cdot(n-1)/2}:\lfloor\log_2 m(A)\rfloor=k\big\}$,
$k=0,\ldots,\lfloor\log_2 n\rfloor$, on which
$\EVec_n$ is computable by Theorem~\ref{t:EV1}.
This time $\{A:\lfloor\log_2 m(A)\rfloor\leq k\}$ 
(rather than ``$\geq k$'') are,
according Observation~\ref{o:MinDim},
effectively open subsets of $\dom(\EVec_n)$.
Hence $U_k$ is relatively open in 
$\{A:\lfloor\log_2 m(A)\rfloor\geq k\}
=\IR^{n\cdot(n-1)/2}\setminus(U_0\cup\cdots\cup U_{k-1})$:
now apply Lemma~\ref{l:Topology}d).
\qed\end{enumerate}\end{proof}
In the discrete realm, the \textsf{Church-Turing Hypothesis} is 
generally accepted and bridges the gap between computational 
practice and formal recursion theory: 
\begin{quote} \it
every function which would naturally be regarded as computable 
is computable under his definition, i.e. by one of his 
{\rm(i.e. Turing's)} machines 
\quad {\rm\cite[p.376]{Kleene}}
\end{quote}
In the real number setting, the Type-2 Machine has not attained
such universal acceptance---mostly due to its inability to
compute any discontinuous function. Hence we propose the 
following as a real counterpart to the discrete Church-Turing Hypothesis:
\begin{quote} \it
The class of real functions $f$
which would naturally be regarded as computable
coincides with those functions 
computable by a Type-2 Machine
\emph{with finite discrete advice of colour classes
in $2\Hausd\KleeneS_1(\dom f)$}.
\end{quote}

\begin{myquestion}
\begin{enumerate}
\item[a)]
Is the rank the (up to permutation) unique
least advice rendering $\LinEq$ computable/continuous?
\item[b)] 
Is the number of distinct eigenvalues the (up to permutation)
unique least advice rendering $\Diag$ computable/continuous?
\item[c)]
More generally, what are sufficient conditions for the sets
$\LEV(f,i)$ ($i=1,\ldots,\Lev(f)$) to be the \emph{unique}
least-size partition of $\dom(f)$ into subsets where
$f$ is continuous?
\end{enumerate}
\end{myquestion}
Recall that in the proof of Corollary~\ref{c:Comparison2},
we have repeatedly employed Lemma~\ref{l:Topology}d) 
giving a sufficient condition for the sets $\LEV(f,i)$ 
to constitute a \emph{least-size} partition of $\dom(f)$ 
into subsets where $f$ is continuous.


\end{document}